\def\draft{0}
\def\anon{0}
\def\lncs{0} 
\def\iacrcc{0}
\def\lncsORiacrcc{0} 

\ifnum\lncs=1
\documentclass[runningheads]{llncs}
\else
    \ifnum\iacrcc=1
    \documentclass[version=submission]{iacrcc}
    \else
    \documentclass[11pt]{article}
    \fi
\fi

\usepackage{macros}
\usepackage{specific-macros}
\usepackage{tikz}
\usepackage{float}
\usetikzlibrary{arrows,shapes}
\usepackage{ifthen}
\usepackage{multirow}

\ifnum\lncsORiacrcc=0
\usepackage{microtype}      
\fi

\title{Adaptive Robustness of Hypergrid Johnson-Lindenstrauss}
\ifnum\anon=0
\author{
Andrej Bogdanov\thanks{University of Ottawa. \href{mailto:abogdano@uottawa.ca}{abogdano@uottawa.ca}. Work supported by an NSERC Discovery Grant.} \and
Alon Rosen\thanks{Bocconi University. \href{mailto:alon.rosen@unibocconi.it}{alon.rosen@unibocconi.it}. Work supported by European Research Council (ERC) under the EU’s Horizon 2020 research and innovation programme (Grant agreement No. 101019547) and Cariplo CRYPTONOMEX grant.} \and
Neekon Vafa\thanks{MIT. \href{mailto:nvafa@mit.edu}{nvafa@mit.edu}. Work supported in part by NSF DGE-2141064 and the grants of the fourth author. Part of this work was done while visiting Bocconi, supported by European Research Council (ERC) under the EU’s Horizon 2020 research and innovation programme (Grant agreement No. 101019547).}\and 
Vinod Vaikuntanathan\thanks{MIT. \href{mailto:vinodv@mit.edu}{vinodv@mit.edu}. Research supported in part by DARPA under Agreement Number HR00112020023, NSF CNS-2154149, a Simons Investigator award and a Ford Foundation Chair at MIT.}}
\else
\ifnum\iacrcc=1
\else
\author{Anonymous submission}
\fi 
\ifnum\lncs=1
\institute{}
\fi 
\fi 
\date{}

\begin{document}

\maketitle

\begin{abstract}
Johnson and Lindenstrauss (Contemporary Mathematics, 1984) showed that for $n > m$, a scaled random projection $\mathbf{A}$ from $\R^n$ to $\R^m$ is an approximate isometry on any set $S$ of size at most exponential in $m$.  If $S$ is larger, however, its points can contract arbitrarily under $\mathbf{A}$. In particular, the hypergrid $([-B, B] \cap \mathbb{Z})^n$ is expected to contain a point that is contracted by a factor of $\kappa_{\mathsf{stat}} = \Theta(B)^{-1/\alpha}$, where $\alpha = m/n$.
 
We give evidence that finding such a point exhibits a statistical-computational gap precisely up to $\kappa_{\mathsf{comp}} = \widetilde{\Theta}(\sqrt{\alpha}/B)$.
On the algorithmic side, we design an online algorithm achieving $\kappa_{\mathsf{comp}}$, inspired by a discrepancy minimization algorithm of Bansal and Spencer (Random Structures \& Algorithms, 2020). On the hardness side, we show evidence via a multiple overlap gap property (mOGP), which in particular captures online algorithms; and a reduction-based lower bound, which shows hardness under standard worst-case lattice assumptions. 

As a cryptographic application, we show that the rounded Johnson-Lindenstrauss embedding is a robust property-preserving hash function (Boyle, Lavigne and Vaikuntanathan, TCC 2019) on the hypergrid for the Euclidean metric in the computationally hard regime.
Such hash functions compress data while preserving $\ell_2$ distances between inputs up to some distortion factor, with the guarantee that even knowing the hash function, no computationally bounded adversary can find any pair of points that violates the distortion bound.
\end{abstract}

\ifnum\iacrcc=0
\thispagestyle{empty}
\thispagestyle{empty}
\fi
\ifnum\lncs=0
 \newpage 
 \tableofcontents
 \thispagestyle{empty}
\fi
\newpage
\ifnum\iacrcc=0
\setcounter{page}{1}
\fi

\def\bbR{\mathbb{R}}

\section{Introduction}

The celebrated Johnson-Lindenstrauss (henceforth JL) lemma \cite{johnson1984extensions,DBLP:conf/stoc/IndykM98} gives us a powerful dimension-reduction mechanism for data.\footnote{The original JL work used a matrix whose rows are unit norm and orthogonal to each other. Indyk and Motwani~\cite{DBLP:conf/stoc/IndykM98} observed that a Gaussian matrix suffices.} The JL lemma states that for all fixed (small) finite sets $S \subseteq \R^n$, for a random i.i.d. Gaussian matrix $\matA \sim \Norm(0,1)^{m \times n}$, the linear map $\frac{1}{\sqrt{m}} \cdot \matA$ embeds $S$ into $\R^m$ in a way that approximately preserves all $\ell_2$ norms. More concretely, the guarantee that 
\[ \Pr_{\matA \sim \Norm(0, 1)^{m \times n}}\left[ \forall \vecx \in S,\;\norm{\matA \vecx}_2 \in (1\pm \epsilon)\cdot \sqrt{m} \cdot \norm{\vecx}_2 \right] \geq 2/3, \]
can be achieved when $m = \Omega(\log(|S|)/\epsilon^2)$. The JL lemma has seen a great deal of work in the mathematics and TCS literature, and has been extended in several directions including faster  and more space-efficient versions~\cite{DBLP:journals/siamcomp/AilonC09,DBLP:conf/soda/KaneN12,DBLP:conf/soda/CohenJN18,jain2022fast}, stronger guarantees~\cite{DBLP:conf/stoc/NarayananN19}, and a proof of optimality~\cite{DBLP:conf/icalp/LarsenN16,DBLP:conf/focs/LarsenN17}.

The statement of the JL lemma crucially relies on the fact the set $S$ is defined independently of the matrix $\matA$. For example, even considering only singleton sets $S = \{\vecx\}$, one can ask whether the order of quantifiers can be switched so that $\vecx$ can be chosen adaptively based on $\matA$, namely:
\[ \Pr_{\matA \sim \Norm(0, 1)^{m \times n}}\left[ \forall \vecx \in \R^n,\;\norm{\matA \vecx}_2 \approx \sqrt{m} \cdot \norm{\vecx}_2 \right] \geq^? 2/3. \]
However, one immediately observes that this is impossible in a very strong sense. For $n > m$, for any matrix $\matA$, one can always find a non-zero vector $\vecx \in \ker(\matA)$, making $\norm{\matA \vecx}_2 = 0$ while $\norm{\vecx}_2$ can be arbitrarily large.
Thus, an adaptive choice of the set of points, one that depends on the dimensionality reduction matrix, kills all nice guarantees that JL gave us. 

One does not have to look too far to see the plausibility of such a scenario. If the matrix $\matA$ is chosen once and for all  (such as in derandomized versions of JL~\cite{DBLP:conf/soda/EngebretsenIO02}), and made public, adversarial entities can choose a set of points that violates the correctness of the JL lemma. Less obviously, even if $\matA$ is not public and only very limited access to it is available, an adversary can reconstruct it and use it to mount the above attack, as shown first by Hardt and Woodruff~\cite{DBLP:conf/stoc/HardtW13,gribelyuk2025lifting}.
More generally, such {\em adaptive} attacks have been extensively in the context of data structures, streaming, property testing, especially in the last decade~\cite{DBLP:conf/stoc/MironovNS08,DBLP:journals/corr/abs-2502-05723,DBLP:conf/focs/GribelyukLWYZ24,DBLP:journals/algorithmica/AttiasCSS24,DBLP:journals/corr/abs-2411-06370,DBLP:conf/nips/NissimST23,DBLP:conf/icml/Cohen0NSSS22,DBLP:conf/stoc/BeimelKMNSS22,DBLP:conf/nips/HassidimKMMS20,beneliezer2023propertytestingonlineadversaries,DBLP:conf/stoc/AlonBDMNY21,DBLP:conf/pods/Ben-EliezerY20}. 
\subsection{The Contracting Hypergrid Vector Problem}

Faced with this pessimistic scenario, we ask whether we can recover the guarantees of the JL lemma if we constrain the set of points $S$, and restrict to {\em resource-bounded} adaptive adversaries. 

One way in which one could constrain $S$, is to zero-out the ``most significant bits'' of the coordinate vectors $\vecx$, i.e., to limit them to some $\ell_{\infty}$ ball. However, the above kernel strategy still works by rescaling $\vecx$ down to live in the ball.

Another option would be to zero-out the ``least significant bits'' of $\vecx$ by requiring $\vecx$ to live in a discrete set, e.g. $\vecx \in \Z^n$. But then, if we have no upper-bound on $\vecx$, we can scale up kernel vectors arbitrarily high so that they in fact live on the integer grid. 

Computationally interesting phenomena occur when combining both these constraints. More precisely, we will require that for all $i \in [n]$, $|x_i| \leq B$ {\em and} $x_i \in \Z$, or more concisely,
\[ \vecx \in \left( [-B, B] \cap \Z \right)^n\]
for some bound $B \in \N$, which could be as small as $1$ or polynomially large in $n$. Phrased another way, this puts a fixed bound on the precision of $\vecx$. Like the kernel examples above, one can ask whether contraction occurs for this {\em hypergrid variant of JL}. We can phrase the problem as follows.
\begin{definition}[Contracting Hypergrid Vector]\label{main-def-intro}
    For $n, m, B \in \N$ and $\kappa \in \R_{>0}$, we define the contracting hypergrid vector ($\problem$) problem with parameters $n, m, B$, and $\kappa$ as follows. Given as input $\matA \sim \Norm(0,1)^{m \times n}$, a valid solution is some $\vecx \in ([-B, B] \cap \Z)^n$ such that
\begin{equation*}\norm{\matA \vecx}_2 < \kappa \cdot \norm{\vecx}_2\cdot  \sqrt{m}.
    \end{equation*}
\end{definition}

Here, the parameter $\kappa$ quantifies the quality of the solution. A direct computation shows that any non-zero choice of $\vecx$ achieves $\kappa = \Theta(1)$ in expectation.
Statistically, the threshold is
\[ \kappa_{\mathsf{stat}} = \Theta\left( (2B+1)^{-n/m} \right).\]
That is, for $\kappa \gg \kappa_{\mathsf{stat}}$, solutions exist in expectation, and for $\kappa \ll \kappa_{\mathsf{stat}}$, they do not.

The story changes if we \emph{computationally bound} the task of finding $\vecx \in \left( [-B, B] \cap \Z \right)^n$ that violate JL. That is, we can consider $\problem$ as a computational problem to be solved by polynomial time algorithms. It is then only natural to ask what the computational complexity of $\problem$ is, and whether it exhibits a statistical-computational gap.

\subsection{Our Results}

If we let $\kappa_{\mathsf{comp}}$ denote the best-possible {\em efficiently achievable} $\kappa$, we show, in a sense to be made precise below, that
\[ \kappa_{\mathsf{comp}} = \widetilde{\Theta}\left(\frac{1}{B} \sqrt{\frac{m}{n}} \right), \]
where $\widetilde{\Theta}$ hides logarithmic terms in all parameters. It is useful to look at these values in terms of the \emph{aspect ratio} $\alpha$, defined as $\alpha := m/n < 1$. In this language, we have
\[ \kappa_{\mathsf{stat}} = \Theta\left( (2B+1)^{-1/\alpha} \right),\;\;\;\;\kappa_{\mathsf{comp}} = \widetilde{\Theta}\left(\frac{\sqrt{\alpha}}{B} \right). \]
To demonstrate how large the statistical-computational gap is, considering only $B=1$ gives a statistical bound that decays exponentially in $1/\alpha = n/m$ as opposed to computationally, where it decays polynomially. We illustrate with a phase diagram in \Cref{fig:phase-diagram}. 

To establish the above value of $\kappa_{\mathsf{comp}}$,
we give two algorithms for $\problem$. One algorithm is a simple variant of the kernel attack described above, while the other is an online algorithm inspired by a discrepancy minimization algorithm of Bansal and Spencer~\cite{BS20}.

We then give matching lower bounds. One is via the multiple overlap gap property, which in particular captures online algorithms. The other is a reduction-based lower bound which shows hardness under standard worst-case lattice assumptions. 

Finally, we show a positive use of the statistical-computational gap by giving a cryptographic application: a construction of \emph{robust property preserving hash functions} for the Euclidean metric. These hash functions compress data while preserving $\ell_2$ distances between input points up to some distortion factor, with the guarantee that no computationally bounded adversary can find any points that violate the distortion bound (even though they must exist).

Robust property preserving hash functions imply collision resistance, demonstrating that the  statistical-computational gap for $\problem$ yields cryptography beyond the existence of one-way functions.

We elaborate on our algorithms, hardness, and applications results in the upcoming sections.

\begin{figure}[t]
\centering
\begin{tikzpicture}[rotate around x=-60]

\begin{axis}[
  clip = false,
  axis lines*=left,
  width=12cm, height=8cm,
  xtick={1, 2, 3, 4}, 
  xticklabels={,,,},
  ytick={0, 1},
  yticklabels={,},
  ztick={1,2,3,4,5,6,7,8,9,10,11,12,13,14,15,16,17,18},
  zticklabels={,,,,,,,,,,,,,,,,,},
  xlabel={$\ln B$},ylabel={$\ln \alpha^{-1}$}, zlabel={$\ln \kappa^{-1}$},
]
\addplot3 [
domain = 1:4,
domain y = 0:1.5,
samples = 10,
samples y = 8,
surf,
fill=blue!50,
faceted color = lightgray] {.5 * y + x};
\addplot3 [
domain = 1:4,
domain y = 0:1.5,
samples = 10,
samples y = 8,
surf,
fill=red!50,
faceted color = red!20] {exp(y) * x};

\node [rotate=60] at (axis cs:1, .6, 10) {no solution};

\node [rotate=35] at (axis cs:4, 1.5, 10) {hard};

\node [rotate=27.5] at (axis cs:4.25, 1.5, 3.5) {easy};
\end{axis}
\end{tikzpicture}
\caption{Phase diagram of $\problem$
in the asymptotic regime (up to lower-order additive terms).  The blue and red boundaries are
$\ln \kappa^{-1} = \tfrac12 \ln \alpha^{-1} + \ln B$ and $\ln \kappa^{-1} = \alpha^{-1} \ln B$, respectively.}
\label{fig:phase-diagram}
\end{figure}
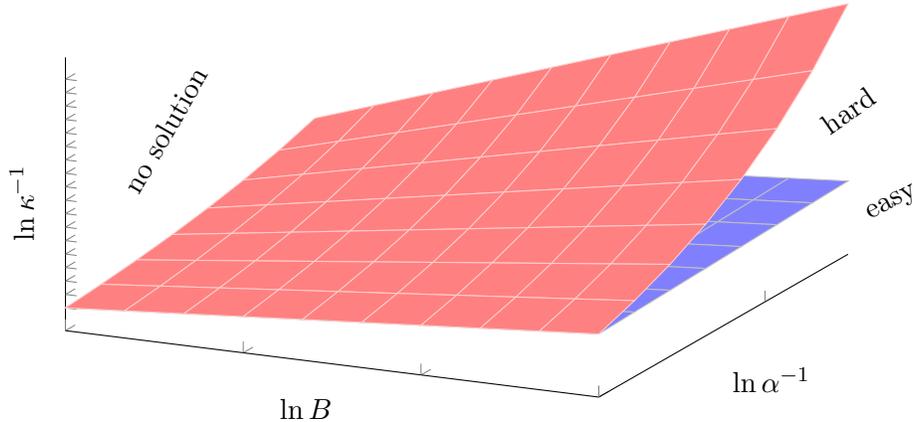

\subsection{Algorithms} 
We show two algorithms for $\problem$. The first uses the same kernel strategy as in the attack against the non-discretized version but is analyzed carefully; and the second is a novel variant of the Bansal-Spencer online discrepancy algorithm~\cite{BS20}. Here, ``online'' means that the algorithm receives every column of $\matA$ one at a time, and after seeing column $j \in [n]$, the algorithm must commit to some choice $x_j \in [-B, B] \cap \Z$.

\begin{theorem}[Informal Version of Theorem~\ref{thm:roundkernel}]\label{thm:informal-roundkernel}
For all $n > m$ and $\matA \sim \Norm(0,1)^{m \times n}$, scaling and rounding a random vector in $\ker(\matA)$ yields $\vecx \in ([-B, B] \cap \Z)^n$ such that $$\frac{\norm{\matA\vecx}_2}{\norm{\vecx}_2} = O\left( \frac{1}{B} \cdot \sqrt{m \log B}\right)$$ with probability $1-o(1)$. This directly solves $\problem$ where
\[ \kappa = O\left( \frac{\sqrt{\log B}}{B} \right) \]
with probability $1 - o(1)$.
\end{theorem}

\begin{theorem}[Informal Version of Theorem~\ref{thm:cool}]\label{informal-cool}
There is an online algorithm for $\problem$ and a universal constant $K$ such that as long as $n \geq K m \log B$, the algorithm achieves $$\kappa = O\left( \frac{1}{B} \cdot \sqrt{\frac{m}{n}} \right)$$ except with probability at most $O(\log B) \cdot  2^{-\Omega(m)}$.
\end{theorem}

Taken together, we have the following corollary, which applies for all ranges of $n > m$.

\begin{corollary}\label{alg-cor-in-intro}
    There is an algorithm for $\problem$ such that for all $n > m$ and all $B$, the algorithm achieves
    \[ \kappa = O\left(\frac{\log B}{B} \cdot \sqrt{\frac{m}{n}} \right) \]
    except with probability at most $O(\log B) \cdot  2^{-\Omega(m)} + o(1)$.
\end{corollary}

\begin{proof}[Proof of \Cref{alg-cor-in-intro} given \Cref{informal-cool,thm:informal-roundkernel}]
    If $n \geq Km \log B$ for the constant $K$ given in \Cref{informal-cool}, apply \Cref{informal-cool}. Otherwise, apply \Cref{thm:informal-roundkernel}.
\end{proof}

\subsection{Hardness} We show two flavors of hardness. The first shows an overlap gap property~\cite{Gamarnik_2021} for the problem, giving evidence that a family of algorithms including local algorithms will fail to solve the problem,  and the second shows computational hardness under the worst-case hardness of lattice problems, using ideas from \cite{DBLP:vv25}.

The multiple overlap gap property ($r$OGP) forbids the existence of certain configurations of multiple solutions $\vecx_1, \dots, \vecx_r$ to a search problem.  It is viewed as an impediment for ``stable'' average-case algorithms in which changing a single input cannot have a large effect on the output.  In our setting, it posits that not all relative angles $\angle \vecx_i, \vecx_j$ can fall into some proscribed range $(\theta_-, \theta_+)$ that does not vanish with $n$.

\begin{theorem}[Informal Version of \Cref{thm:ogp}]\label{thm:informal-ogp}
Assuming $n \geq m$ and
\[\kappa \ll \frac{1}{B} \cdot  \frac{1}{\sqrt{\log(n/m)}} \cdot \sqrt{\frac{m}{n}},\] most instances of $\problem$ satisfy $r$OGP for sufficiently large $n$.
\end{theorem}

In \Cref{thm:onlinehard}, we derive hardness of online algorithms in the same parameter regime from OGP hardness.  Our result is not fully general, in the sense that we assume the algorithm is committed to the approximate norm of the solution $\vecx$ (before seeing the input $\matA$).  Our online algorithm, as well as the ones of Bansal and Spencer~\cite{BS20}, satisfy this assumption.  We believe that some assumption of this type is necessary for an OGP-based argument to ensure stability. In general, we conjecture that no online algorithm can succeed when $\kappa$ is at most $(\sqrt{\pi/8} - o(1))\sqrt{\alpha}/B$ (\Cref{online threshold conjecture}).

The lattice-based lower bound applies to the parameter regime in which $\kappa$ vanishes as $n$ grows.

\begin{theorem}[Informal Version of \Cref{lattice-based-lower-bound}]\label{thm:informal-lattice-lb}
Assuming the polynomial hardness of standard worst-case lattice problems, for all $B, n \leq \poly(m)$, there does not exist any polynomial time algorithm for $\problem$ satisfying
\[ \kappa \ll \frac{1}{B} \cdot \frac{1}{\sqrt{n}}. \]
\end{theorem}
We emphasize that \Cref{thm:informal-lattice-lb} shows hardness of an \emph{average-case} problem (namely, $\problem$) assuming only the \emph{worst-case} hardness of lattice problems. We note however that this lower bound is quantitatively weaker than the OGP analysis, and is meaningful for a somewhat more restricted range of parameters. For example, when $B = \Theta(1)$, we need $n \gg m \log m$ for a solution to exist for $\kappa \approx 1/(B \sqrt{n})$. Nonetheless, this lower bound is still much higher than $\kappa_{\mathsf{stat}}$ for a wide range of parameters. We leave it as a fascinating open question as to whether this lower bound can be improved to match that of \Cref{thm:informal-ogp} or similar.

\subsection{Robust Locality Sensitive Hashing}
Semantic embeddings, which compress data points into vectors whose (Euclidean, say) distance approximates some semantic distance between data points, are a recent and very popular paradigm in machine learning~\cite{DBLP:journals/corr/abs-1301-3781,DBLP:journals/corr/abs-1908-10084,radford2021learningtransferablevisualmodels,muennighoff2023mtebmassivetextembedding,hugging}. 

Several recent works~\cite{DBLP:conf/emnlp/SongRS20,DBLP:conf/nips/TongJS23,DBLP:conf/uss/ZhangJBS24} point out the issue of adversarial robustness, under the term ``adversarial semantic collisions'', namely, inputs (text, images, and so on) that are semantically different yet hash to the same (or close) outputs. In light of this real-world phenomenon, the question of whether one can design semantically collision-resistant hash functions, ones for which semantic collisions may exist but are computationally hard to find, is practically important. 

While we do not solve this problem, we suggest a possible approach. First, we observe that the computational-statistical gap for $\problem$ allows us to design robust locality-sensitive hash functions~\cite{DBLP:conf/innovations/BoyleLV19} for the {\em Euclidean distance}. These compressing hash functions preserve the $\ell_2$ distance between input points up to some distortion factor $\xi$, assuming the points are chosen by a computationally bounded adversary. (See \Cref{def-of-our-hash-primitive} for a formal definition.) Secondly, if one could design a {\em collision-free, possibly dimension-expanding} embedding from ``semantic distance'' to Euclidean distance, we can compose the two to get a compressing semantically collision-resistant hash function. 

We do not pursue the second of these problems in this paper; rather, we restrict our attention to designing a robust locality-sensitive hash function for the Euclidean distance.

\begin{theorem}[Informal Version of~\Cref{main-construction}]
    Suppose that $\problem$ is hard for parameters $n,m,B$, and $\kappa$. Then, there is a robust locality sensitive hash function for the Euclidean norm mapping the domain $([0, B] \cap \Z)^n$ into $\R^m$ with distortion
    \[ \xi =  O\left(\frac{1}{\kappa} \sqrt{\frac{n}{m}} \right) \]
    that is compressing as long as
       \[ (B+1)^n \gg \left( \frac{Bn}{\kappa \sqrt{m}} \right)^m. \]
\end{theorem}

The construction is a variant of the Johnson-Lindenstrauss embedding where the output gets rounded to some grid $\gamma \Z^m$. Phrased another way, we show that the rounded Johnson-Lindenstrauss embedding is itself a robust locality sensitive hash function for the Euclidean norm where the domain is a hypergrid. We give the explicit construction in \Cref{fig:construction}.

\begin{corollary}[Informal]
    Suppose that the online algorithm in \Cref{informal-cool} is optimal for $\problem$. Then, there is a compressing robust locality sensitive hash function for the Euclidean norm mapping the domain $([0,B] \cap \Z)^n$ into $\R^m$ with the following parameters: 
    \begin{itemize}
        \item For $B = \Theta(1)$, as long as $n \gg m \log m$, the distortion $\xi$ can be as low as $\approx n/m$.
        \item For $B = n^{\Theta(1)}$, as long as $n \gg m$, the distortion $\xi$ can be as low as $\approx Bn/m$.
    \end{itemize}
\end{corollary}
We emphasize these distortion bounds are only meaningful up to $B \sqrt{n}$, as any distinct elements of $([0, B] \cap \Z)^n$ have $\ell_2$ distance between $1$ and $B \sqrt{n}$.

For a concrete application of these hash functions, consider (e.g., grayscale) images with pixel values in $\{0, 1, \cdots, 255\}$, where $\ell_2$ distance encodes some semantic information. Our hash function generically compresses such images (with $B = 255$) while preserving approximate $\ell_2$ norms of images and approximate $\ell_2$ distance between images (to a computationally bounded adversary). 

As another simple corollary of the construction, we note that our definition of robust locality sensitive hash functions implies collision resistance, whose existence is known to be stronger than the existence of one-way functions~\cite{DBLP:conf/eurocrypt/Simon98}. Therefore, in particular, the statistical-computational gap of $\problem$ yields cryptographic utility beyond the existence of one-way functions.

\subsection{Related Work}

\paragraph{$\problem$, $\mathsf{SBP}$ and $\mathsf{NBV}$.} Two problems closely related to $\problem$ are the Symmetric Binary Perceptron (SBP) problem, introduced by Aubin, Perkins, and Zdeborov\'{a}~\cite{Aubin_2019} and the Nearest Boolean Vector (NBV) problem, introduced by Mohanty, Raghavendra, and Xu~\cite{mohanty-raghavendra-xu}.  The differences from $\problem$ are as follows:
\begin{itemize}
    \item In SBP and NBV, the domain of $\vecx$ is fixed to $\{-1, 1\}^n$, as opposed to $([-B, B] \cap \Z)^n$.  In particular, the norm of $\vecx$ is fixed and $\vecx$ cannot have zero entries.
    \item SBP asks for a bound on $\norm{\matA \vecx}_{\infty}$, as opposed to $\norm{\matA \vecx}_2$.  (Therefore there is no normalizing $\sqrt{m}$ term.)
\end{itemize}

The search variant of SBP has been studied extensively for its statistical and computational thresholds~\cite{Aubin_2019,BDVLZ20,DBLP:conf/stoc/PerkinsX21,abbe2021proofcontiguityconjecturelognormal,DBLP:conf/stoc/AbbeLS22,DBLP:conf/focs/GamarnikK0X22,DBLP:conf/colt/GamarnikK0X23,Barbier_2024}.  Both the algorithm of Bansal and Spencer and the $r$OGP bound of Gamarnik, K{\i}z{\i}lda\u{g}, Perkins and Xu~\cite{DBLP:conf/focs/GamarnikK0X22} match  ours in the special case $B = 1$.  

The study of NBV has focused on refuting proximity to a random subspace in the unsatisfiable regime~\cite{ghosh-etal, potechin-turner-venkat-wein, bogdanov-rosen}.

\paragraph{Adaptively Robust $X$.} Adaptive robustness has been studied in many related contexts, perhaps stemming from the work of Mironov, Naor and Segev~\cite{DBLP:conf/stoc/MironovNS08} in the context of sketching algorithms. Many recent works have further explored this question in the context of sketching and streaming algorithms~\cite{DBLP:journals/corr/abs-2502-05723,DBLP:conf/focs/GribelyukLWYZ24,DBLP:journals/algorithmica/AttiasCSS24,DBLP:journals/corr/abs-2411-06370,DBLP:conf/icml/Cohen0NSSS22,DBLP:conf/stoc/BeimelKMNSS22,DBLP:conf/nips/HassidimKMMS20}, randomized data structures~\cite{naor2019bloomfiltersadversarialenvironments}, property testing~\cite{beneliezer2023propertytestingonlineadversaries}, online algorithms~\cite{DBLP:conf/stoc/AlonBDMNY21}, and sampling algorithms~\cite{DBLP:conf/pods/Ben-EliezerY20}.

Gribelyuk, Lin, Woodruff, Yu, and Zhou~\cite{gribelyuk2025lifting} give an efficient algorithm that, for any sufficiently compressing unknown linear embedding $\matA$, finds a hypergrid vector $\vecx$ that fails to embed almost-isometrically  under $\matA$ given only query access to $\matA$.  In particular, their result implies an easy regime for $\problem$ when $\kappa$ is close to one.  In contrast, our algorithms apply to much smaller values of $\kappa$ (and are thus stronger) but are specific for the Johnson-Lindenstrauss embedding and assume unrestricted access to $\matA$.

\section{Technical Overview}

\subsection{Algorithms}

We focus on the online algorithm. The algorithm iterates over $t \in \{1, 2, \dots, n\}$ in order, receives each column $\veca_t \sim \Norm(0,1)^m$ one at at time, and commits to $x_t \in [-B, B] \cap \Z$ before incrementing $t$.

At first, we describe a simpler variant of the algorithm that produces $x_t \in \{-1, 1\}$ (which is a stronger constraint, but will yield a weaker bound). This in particular ensures that $\norm{\vecx}_2$ is fixed at $\sqrt{n}$, so the goal of the algorithm is simply to minimize $\norm{\matA \vecx}_2$. At any given point in time $t \in [n]$, there exists a current state
\[\vecy_t =\sum_{i = 1}^{t-1} x_i \veca_i \in \R^m,\]
corresponding to the result of the choices it has made so far.
The algorithm must choose $\vecy_{t+1}$ as $\vecy_{t+1} = \vecy_t + b \veca_t$ for some $b \in \{\pm 1\}$. By rotational symmetry of the Gaussian, it is clear that the optimal online choice is to choose
\[ x_t = \argmin_{b \in \{-1 , 1\}} \norm{\vecy_t + b \veca_t}_2, \]
as the only thing that matters about $\vecy_{t+1}$ is its $\ell_2$ norm.

We proceed to analyze this algorithm. We can decompose $\veca_t$ into its perpendicular and parallel components with respect to $\vecy_t$. Explicitly, by spherical symmetry of the Gaussian, we have
\[\veca_t = a^{\parallel} \cdot \frac{\vecy_t}{\norm{\vecy_t}_2} + \veca^{\perp}, \]
where $a^{\parallel} \sim \Norm(0,1)$ and $\veca^{\perp}$ is a spherical multivariate Gaussian on the $(m-1)$-dimensional subspace perpendicular to the line spanned by $\vecy_t$. By the Pythagorean theorem, it follows that
\[ \norm{\vecy_t + b \veca_t}_2^2 = \norm*{\vecy_t + b a^{\parallel} \cdot \frac{\vecy_t}{\norm{\vecy_t}_2} + b\veca^{\perp}}_2^2 = \norm*{\vecy_t + \frac{b a^{\parallel}}{\norm{\vecy_t}_2} \cdot \vecy_t}_2^2 + \norm{\veca^{\perp}}_2^2 . \]
Choosing $b = -\sign(a^{\parallel})$ minimizes this quantity, which yields
\[ \min_{b \in \{-1, 1\}} \norm{\vecy_t + b \veca_t}_2^2 = \left(\norm{\vecy_t}_2 - \left|a^{\parallel}\right|\right)^2 + \norm{\veca^\perp}_2^2. \]
Letting $R_t = \norm{\vecy_t}_2$ and expanding, we get the stochastic recurrence
\[ R_{t+1}^2 = R_t^2 - 2 R_t |z_1| + \norm{\vecz}_2^2,\]
where $\vecz = (z_1, \dots, z_m) \sim \Norm(0,1)^m$. As the typical value of $|z_1|$ is $\Theta(1)$ and the typical value of $\norm{\vecz}_2^2$ is $\Theta(m)$, we observe that we get negative drift in $R_t$ whenever $R_t \gg m$, in the sense that $R_{t+1} \ll R_t$ with good probability. One can show that the fixed point of this recurrence is $R_t = \Theta(m)$, independent of $t$. This results in $\norm{\matA \vecx}_2 \leq O(m)$ and $\norm{\vecx}_2 = \sqrt{n}$, giving
\[ \kappa = \frac{\norm{\matA \vecx}_2}{\sqrt{m} \cdot \norm{\vecx}_2} \leq O\left( \sqrt{\frac{m}{n}}\right), \]
as desired.

We briefly describe how to reduce $\kappa$ by a factor of $B$ by allowing $x_t \in [-B, B] \cap \Z$. For the first half of the steps (i.e., $t \leq n/2$), we set the ``temperature'' all the way up to $B$, to enforce $x_t \in \{-B, B\}$. This ensures that $\norm{\vecx}_2 \geq \Omega(B \sqrt{n})$, regardless of what happens in the second half of the steps. However, the fixed point of the recurrence becomes $R_{n/2} = \Theta(Bm)$, which defeats the purpose of increasing $\norm{\vecx}_2$. For the second half of the steps, we carefully choose a ``cooling'' schedule (see \Cref{fig:cool}) to get back to temperature $1$, quickly converging to $R_n = \Theta(m)$. This results in $\norm{\matA \vecx}_2 \leq O(m)$ and $\norm{\vecx}_2 \geq \Omega(B \sqrt{n})$, giving 
\[ \kappa = \frac{\norm{\matA \vecx}_2}{\sqrt{m} \cdot \norm{\vecx}_2} \leq O\left(\frac{1}{B}  \sqrt{\frac{m}{n}}\right), \]
as desired, at the cost of requiring $n \geq K m \log B$ for some universal constant $K$. For more details, we defer to \Cref{sec:online-alg}.

\subsection{Hardness}

The multi-OGP hardness (\Cref{thm:ogp}) is derived from a first moment (annealed) estimate of the expected number of forbidden configurations in a random instance.  Our analysis extends that of Gamarnik et al. for the SBP~\cite{DBLP:conf/focs/GamarnikK0X22}, where the solution space is the Boolean cube.  

The main conceptual difference is in the choice of distance metric.  Unlike for the Boolean cube, the Euclidean and Hamming metrics are not equivalent over the hypergrid.  We prove OGP with respect to the normalized inner product.  The utility of this metric is demonstrated in our online lower bound (\Cref{thm:onlinehard}), which essentially shows that it captures the distance between the outputs of executions over
correlated inputs that share the same prefix and independent suffixes.

For the lattice-based lower bound (\Cref{lattice-based-lower-bound}), one approach would be to adapt the recent work of Vafa and Vaikuntanathan~\cite{DBLP:vv25} that shows a reduction from worst-case lattice problems to SBP. Instead of opening up their proof in a white-box way, we choose to reduce from an intermediate problem called ``Continuous Learning With Errors'' (CLWE)~\cite{DBLP:conf/stoc/Bruna0ST21}. This average-case problem, which is known to be as hard as worst-case approximate lattice problems~\cite{DBLP:conf/stoc/Bruna0ST21,DBLP:conf/focs/GupteVV22}, asks to distinguish
\[(\matA, \vecs^\top \matA + \vece^\top \pmod{1}), \text{\ \ \ \  and \ \ \ \ } (\matA, \vecb^\top),\] where $\matA \sim \Norm(0,1)^{m \times n}$, $\vecs \sim n^{\varepsilon} \cdot \mathbb{S}^{m-1}$, $\vece \sim \Norm(0, 1/\poly(n))^n$, and $\vecb$ is a uniformly random vector mod $1$, where $\mathbb{S}^{m-1}$ denotes the unit sphere in $\R^m$. By using integrality of $\vecx$ and simultaneous smallness of $\vecx$ and $\matA \vecx$, we can multiply the second element in the pair on the right by $\vecx$ and check if the result is small modulo $1$. We defer the details to \Cref{sec:lattice-lb}.

\subsection{Robust Locality Sensitive Hashing}
To design robust locality sensitive hash functions in the Euclidean norm, our starting-point is the JL lemma and the syntax of $\problem$. We set the hash function to be
\begin{align*} \hash_{\matA} &: ([0, B] \cap \Z)^n \to \R^m,
\\\vecx &\mapsto \frac{1}{\sqrt{m}} \cdot \matA \vecx,
\end{align*}
where the key of the hash function is the matrix $\matA \sim \Norm(0,1)^{m \times n}$. By linearity and a direct reduction from (the hardness of) $\problem$, it is hard to find $\vecx_1, \vecx_2 \in ([0, B] \cap \Z)^n$ such that
\begin{equation} \label{eqn:contraction-tech-overview} \norm{\hash_{\matA}(\vecx_1) - \hash_{\matA}(\vecx_2)} < \kappa \norm{\vecx_1 - \vecx_2},
\end{equation}
as otherwise, $\vecx_1 - \vecx_2 \in ([-B, B] \cap \Z)^n$ would be a solution to $\problem$. Therefore, we have the guarantee for this function, it is computationally hard to find two points whose distance between the hashes is multiplicatively smaller by a factor of $1/\kappa$.

As is, this construction has two issues:
\begin{enumerate}
    \item \label{item:one-sided} Eq. \eqref{eqn:contraction-tech-overview} is only a one-sided guarantee. How do we ensure that it is hard to find $\vecx_1, \vecx_2 \in ([-B, B] \cap \Z)^n$ such that
    \[ \norm{\hash_{\matA}(\vecx_1) - \hash_{\matA}(\vecx_2)} > \eta \norm{\vecx_1 - \vecx_2}, \]
    for some parameter $\eta \gg 1$?
    \item \label{item:compression-overview} Even if $m < n$, what does it mean for this function to be compressing in a bit-complexity sense if the codomain is $\R^m$?
\end{enumerate}
Thankfully, the solutions to both \Cref{item:one-sided,item:compression-overview} are relatively simple. For \Cref{item:one-sided}, we can use the spectral norm bound that $\norm{\matA}_2 \leq O(\sqrt{n})$ with high probability. This implies that we can set $\eta = O(\sqrt{n/m})$ (since $\matA$ is scaled down by $\sqrt{m}$). This results in overall distortion
\[ \frac{\eta}{\kappa} = O\left( \frac{1}{\kappa} \sqrt{\frac{n}{m}} \right).\]
For \Cref{item:compression-overview}, we can both discretize $\R^m$ into a grid $\gamma \Z^m$ and upper bound its $\ell_2$ norm given that $\norm{\matA}_2 \leq O(\sqrt{n})$ is bounded anyways. Then, we can set the codomain to $\gamma \Z^m \cap \ball_m(r)$ for sufficiently large $r$. Showing compression amounts to counting the number of points in the discretized ball as compared to the the number of points in the domain, $(B+1)^n$. We defer to \Cref{sec:hashing} for deatils.

\section{Preliminaries}
For a natural number $n \in \N$, we let $[n]$ denote the set $\{1, 2, \cdots, n\}$. For real numbers $a, b \in \R$ with $a \leq b$, we let $[a, b]$ denote the continuous interval $\{x \in \R : a \leq x \leq b\}$. We say a function $f : \N \to \R_{>0}$ is negligible if for all $c > 0$, $\lim_{n \to \infty} f(n) \cdot n^c = 0$. We use the notation $\negl(n)$ to denote a function that is negligible (in its input $n$). We similarly use the notation $\poly(n)$ to denote a function that is at most $n^{O(1)}$. As shorthand, we say an algorithm is p.p.t. if it runs in probabilistic polynomial time.

We let $\Norm(\mu, \sigma^2)$ denote the univariate Gaussian (or normal) distribution with mean $\mu \in \R$ and standard deviation $\sigma \in \R_{>0}$. For any distribution $\mathcal{D}$, we let $\mathcal{D}^m$ denote the product distribution of $m$ i.i.d. copies of $\mathcal{D}$. We let $\chi^2_m$ denote the chi-squared distribution with $m$ degrees of freedom, which is defined as the distribution of a random variable $Z$ such that
\[ Z = \sum_{i \in [m]} X_i^2 \]
for i.i.d. $X_i \sim \Norm(0, 1)$. Equivalently, for $\vecv \sim \Norm(0,1)^m$, the distribution of $\norm{\vecv}_2^2$ is identically $\chi^2_m$.

\begin{definition}[Contracting Hypergrid Vector Problem ($\problem$)]\label{main-def}
    For $n, m, B \in \N$ and $\kappa \in \R_{>0}$ with $m < n$, we define the $\problem$ problem with parameters $n, m, B$, and $\kappa$ as follows. Given as input $\matA \sim \Norm(0,1)^{m \times n}$, a valid solution is some $\vecx \in ([-B, B] \cap \Z)^n$ such that
\begin{equation}\label{eq:sbp} \norm{\matA \vecx}_2 < \kappa \norm{\vecx}_2 \sqrt{m}.
    \end{equation}
To match notation in the literature, we use $\alpha$ to denote the aspect ratio $m/n < 1$.
\end{definition}

We say that $\problem$ is (computationally) \emph{hard} for parameters $n, m, B, \kappa$ if for all p.p.t. algorithms $\alg$, the probability that $\alg$ outputs a valid solution to $\problem$ for parameters $n, m, B, \kappa$ is $\negl(n)$.

\section{Algorithms}

We analyze two efficient algorithms for solving $\problem$. Taken together, they give algorithms for $\problem$ whenever 
$\kappa \gg \sqrt{\alpha} \cdot \log B / B$, provided $m < n$.

\subsection{Online Norm Minimization}\label{sec:online-alg}

Our algorithm processes the columns $\veca$ of $\matA$ in sequence producing the corresponding entries of $\vecx$ one by one.  It is inspired by one of the algorithms of Bansal and Spencer~\cite{BS20} where both $\matA$ and $\vecx$ are restricted to $\pm 1$ values and, less importantly, their objective is to minimize the infinity norm of $\matA \vecx$.  

We assume $B$ is a power of two. If not, use the largest available. $K$ is a sufficiently large absolute constant.

    \begin{figure}[H]
    \centering
    \fbox{
    \begin{minipage}{0.95\textwidth}
    \begin{center}
    Online Algorithm $\cool$
    \end{center}
    \begin{pseudocode}
    {\bf State:} $\vecy \in \R^m$, initialized with $\zero$. \\
    {\bf Parameter:} The temperature $b \in [-B, B] \cap \Z$. \\ \\
    {\bf Step:} On sample $\veca \in \R^m$, \\
    \> Update $\vecy$ to $\vecy - b\veca$ or $\vecy + b\veca$, whichever is smaller in 2-norm. \\
    \> Output the minimizer $-b$ or $b$. \\ \\
    {\bf Algorithm} $\cool$:  Run $n$ steps with the temperature $b$ set to: \\
    \> \> \> \> \> $B$ in the first $n - K m(\log(B) - 1) $ steps, \\
    \> \> \> \> \> $B/2$ in the next $Km$ steps, \\
    \> \> \> \> \> $B/4$ in the next $Km$ steps, \\
    \> \> \> \> \> \> \>  $\vdots$ \\
    \>  \> \> \> \>  $1$ in the last $Km$ steps.
    \end{pseudocode}
    \end{minipage}
    }
    \caption{Online Norm Minimization Algorithm $\cool$, as analyzed in \Cref{thm:cool}.}
    \label{fig:cool}
\end{figure}

The only difference between Bansal and Spencer's transition rule and ours is the choice of norm to be minimized.\footnote{Specifically, their ``majority rule'' strategy (Strategy 2) can be interpreted as minimizing the $\ell_1$ norm in the online step~\cite{BS20}.} The 2-norm is more natural in the Gaussian setting.  
Our innovation is the cooling schedule which is responsible for the factor $B$ scaling of the discrepancy.

\begin{theorem}
\label{thm:cool}
Assuming the samples are independent normals and $n \geq 2 K m \log B$, the output $\vecx$ of $\cool$ satisfies $\norm{\matA \vecx}/\norm{\vecx} = O(m/B\sqrt{n})$ except with probability at most $O(\log B) \cdot  2^{-\Omega(m)}$.
\end{theorem}

The norm of $\vecx$ is dominated by the temperature $B$ part of the schedule so it is at least $(B/2)\sqrt{n}$.  It remains to show that the final state has norm at most $O(m)$.

The state update $\vecy' = \vecy + x\veca$ can be decomposed as $\vecy + x \veca^\parallel + x\veca^\perp$, where $\veca^\parallel$ is the component of $\veca$ in the direction of $\vecy$, and $\veca^\perp$ is its orthogonal complement.  By Pythagoras' theorem,
\[ \norm{\vecy'}^2 = \norm{\vecy + x \veca^\parallel}^2 + \norm{x\veca^\perp}^2 = \norm{\vecy + x \veca^\parallel}^2 + b^2\norm{\veca^\perp}^2. \]
As $\vecy$ and $\veca^\parallel$ are aligned, $\norm{\vecy + x \veca^\parallel}$ is either $\norm{\vecy} - b\norm{\veca}$ or $\norm{\vecy} + b\norm{\veca}$.  The first choice is clearly the minimizing one and
\[ \norm{\vecy'}^2 = \bigl(\norm{\vecy} - b\norm{\veca^\parallel}\bigr)^2 + b^2 \norm{\veca^\perp}^2. \]
As the entries of $\veca$ are independent normals, by spherical symmetry, $\norm{\veca^\parallel}$ and $\norm{\veca^\perp}^2$ are the absolute value $\lvert N\rvert$ of a normal random variable and an independent chi-squared random variable $\chi_{m-1}^2$ with $m - 1$ degrees of freedom, respectively.  The length $L = \norm{\vecy}_2$ satisfies the stochastic recurrence
\begin{equation}
\label{eq:process}
L' = \sqrt{(L - b\lvert N\rvert)^2 + b^2 \chi^2_{m-1}},
\end{equation}
where $L'$ denotes the updated length after one step.

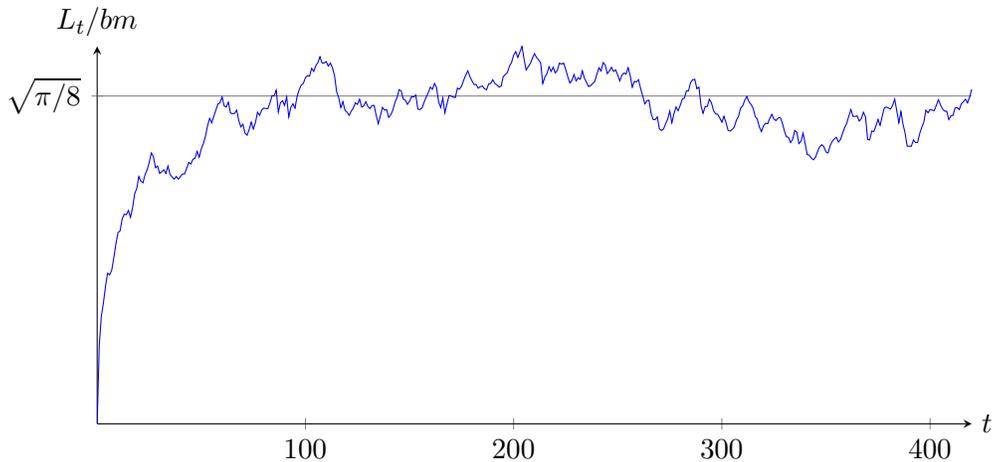
\begin{figure}[ht]
\centering
\begin{tikzpicture}
\begin{axis}[%
width=0.8\textwidth, height=0.4\textwidth,  
axis lines=center,
xtick={0, 100, 200, 300, 400}, ytick={0, .627, .8}, 
xmax=420,
yticklabel={$\sqrt{\pi/8}$}, 
x label style={at={(axis description cs:1,0)}, anchor=west}, 
y label style={at={(axis description cs:0,1)}, anchor=south}, 
xlabel={$t$}, ylabel={$L_t/bm$}]
\addplot[mark=none, gray] coordinates {(0,.627) (420,.627)};
\addplot[mark=none, blue] table [x=t, y=U] {algorithm-run.csv};
\end{axis}
\end{tikzpicture}
\caption{A sample realization of the stochastic process~\eqref{eq:process} with $m = 50$ with fixed temperature $b$. Since $b$ is fixed, $L_t$ is homogeneously linear in $b$, so the stochastic process $L_t/(bm)$ is independent of $b$.}
\end{figure}

Applying the inequality $\sqrt{1 + x} \leq 1 + x/2$ valid for $x \geq 0$ yields
\begin{equation}
\label{eq:drift}
L' \leq \lvert L - b\lvert N\rvert\rvert + \frac{b^2 \chi^2_{m-1}}{2\lvert L - b\lvert N\rvert\rvert}.
\end{equation}
The typical magnitude of $\lvert N\rvert$ is constant and the typical magnitude of $\chi^2_{m-1}$ is on the order of $m$.   If $L$ is larger than about $bm$, the drift is typically negative.

\begin{claim}
\label{claim:drift}
Assume $m \geq 4$. Conditioned on $L \geq 2b m$, $L'$ is stochastically dominated by $L - b\typeB$ where $\typeB$ is independent of $L$, has mean at least $0.25$, and has constant subexponential norm, in the sense that $\Pr[|\typeB| \geq t ] \leq O \left( \exp\left( -\Omega(t) \right) \right)$ for all $t > 0$.
\end{claim}
For now, we defer the proof of \Cref{claim:drift}.

To handle the cases when $L$ is small or $\lvert N\rvert$ is atypically large, we can use the simple bound that $\norm{\vecy'}$ is still at most $\norm{\vecy} + b\norm{\veca}$ by the triangle inequality so
\begin{equation}
\label{eq:worstdrift}
L' \leq L + b\typeA, \qquad\text{where $\typeA$ is of type $\sqrt{\chi^2_m}$}.
\end{equation}

We will refer to the corresponding random variables as in \eqref{eq:worstdrift} and \Cref{claim:drift} as being of type $\typeA$ and type $\typeB$, respectively. By standard properties of the $\chi^2_m$ distribution and Jensen's inequality, type $\typeA$ is of mean at most $\sqrt{m}$ and has subgaussian norm $\Omega(\sqrt{m})$.

Let $L_0, L_1, \dots$ be a stochastic process that evolves according to~\eqref{eq:process} (with fixed $b$).

\begin{claim}
\label{claim:evolution}
If $L_0$ is at most $8bm$ then for every $t^\star \geq Km$, $L_{t^\star}$ is at most $4bm$ except with probability $O(2^{-\Omega(m)})$.
\end{claim}
\begin{proof}
Let $A_1, A_2, \dots$ and $B_1, B_2, \dots$ be sequences of independent type-$\typeA$ and type-$\typeB$ random variables, respectively.  The drift $L_t - L_{t-1}$ is stochastically dominated by $bA_t$ if $L_{t-1} < 2bm$ and by $-bB_t$ otherwise.

Let $T$ be the last time $T \leq Km$ at which $L_T < 2bm$, if such a time exists. By~\Cref{claim:drift} and~\eqref{eq:worstdrift}, $L_{t^\star}$ is then stochastically dominated by $2bm + b(A_T - B_{T+1} - \cdots - B_{t^\star})$.  Otherwise, $L_t$ is stochastically dominated by $L_0 - bB_1 - \cdots - bB_{t^\star}$.  By a union bound over the possible choices of $T$,
\begin{align*}
\Pr[L_{t^\star} > 4bm] &\leq \Pr[L_0 - bB_1 - \cdots - bB_{t^\star} > 4bm] + \sum_{t=1}^{t^\star} \Pr\left[ bA_t - bB_{t+1} - \cdots - bB_{t^\star} > 2bm \right] \\
&\leq \Pr\left[B_1 + \cdots + B_{t^\star} \leq 4m \right] + \sum_{t=0}^{\infty} \Pr\left[A_1 - B_1 - \cdots - B_t > 2m\right].
\end{align*}
As the type-$\typeB$ random variables have mean at least $0.25$ and are subexponential, the first probability is at most $2^{-m}$ provided $t^* > Km$. As for the $t$th term in the sum, by union and tail bounds,
\begin{align*}
\Pr\left[ A_1 - B_1 - \cdots - B_t > 2m \right]
&\leq \Pr\left[ A_1 > m + t/8 \right] + \Pr\left[B_1 + \cdots + B_t < -m + t/8\right] \\
&\leq \exp \left( -\Omega \left(\sqrt{m} + t/\sqrt{m} \right)^2 \right) + \exp \left( -\Omega(m + t) \right) \\
&\leq 2 \exp\left( -\Omega(m + t) \right). 
\end{align*}
By convergence of the geometric series, the sum over $t$ is also bounded by $O\left(\exp \left( -\Omega(m) \right) \right)$.
\end{proof}

\begin{proof}[Proof of \Cref{thm:cool}]
Apply Claim~\ref{claim:evolution} to $L = \norm{\vecy}$. At the end of the first stage $L$ is at most $4Bm = 8(B/2)m$ except with probability $O(2^{-\Omega(m)})$.  Assuming it is, at the end of the second stage $L$ is at most $4(B/2)m = 8(B/4)m$ except with probability $O(2^{-\Omega(m)})$, and so on.  At the very end $L$ is at most $4m$ as desired. The cumulative failure probability is at most $O(\log B)2^{-\Omega(m)}$.
\end{proof}

We now proceed to prove~\Cref{claim:drift}.
\begin{proof}[Proof of Claim~\ref{claim:drift}]
Set
\[ \typeB = \begin{cases}
\lvert N\rvert - \chi^2_{m-1}/(2m), &\text{if $\lvert N\rvert \leq m$} \\
-\sqrt{m^2 + \chi^2_{m-1}}, &\text{otherwise}.
\end{cases}
\]
Stochastic domination follows from~\eqref{eq:drift} and~\eqref{eq:worstdrift}, as we can decompose the type $\typeA$ random variable into its component in the direction of $\vecy$ (i.e., $N$) and its $m-1$ independent other components. By standard facts and Cauchy-Schwarz, the mean of $\typeB$ is at least
\begin{align*}
\E[\typeB] &= \E[\lvert N \rvert] - \frac{\E\left[\chi_{m-1}^2\right]}{2m}  - \E \left[\left(
\sqrt{m^2 + \chi_{m-1}^2} + \lvert N \rvert - \frac{\chi^2_{m-1}}{2m} \right) \one\left[ \lvert N\rvert > m \right] \right]  \\
&\geq \sqrt\frac{2}{\pi} - \frac{m-1}{2m}
- \left(\sqrt{m^2 + \E\left[\chi_{m-1}^2 \right]} + \sqrt{\E\left[ N^2 \right]}  + \sqrt{ \frac{\E \left[ \chi_{m-1}^4 \right]}{4m^2}}\right) \sqrt{\Pr \left[ \lvert N \rvert > m \right]} \\
&\geq 0.29 - (2m + 2) \cdot 2\exp\left(-m^2/4 \right) \\
&\geq 0.25.
\end{align*}
For the subexponential norm, by union and large deviation bounds,
\begin{align*}
\Pr\left[ \lvert \typeB\rvert \geq t \right] &\leq \Pr\left[ \lvert N\rvert \geq t\right] + \Pr\left[\frac{\chi_{m-1}^2}{2m} \geq t\right] + \Pr\left[\text{$\sqrt{m^2 + \chi_{m-1}^2} \geq t$ and $\lvert N\rvert > m$}\right] \\
&\leq 2\exp\left(-t^2/2 \right) + \exp\left(-t/2 \right) +  \Pr\left[\text{$\sqrt{m^2 + \chi_{m-1}^2} \geq t$ and $\lvert N\rvert > m$}\right] \\
&= O(\exp(-t/2)) +  \Pr\left[\text{$\sqrt{m^2 + \chi_{m-1}^2} \geq t$ and $\lvert N\rvert > m$}\right].
\end{align*}
To bound the right-hand term, we consider two cases. If $t < m^2$, then
\[ \Pr\left[ \text{$\sqrt{m^2 + \chi_{m-1}^2} \geq t$ and $\lvert N\rvert > m$} \right] \leq \Pr[|N| > m] \leq O \left( \exp\left(-m^2/4 \right) \right) \leq O(\exp(-\Omega(t))).\]
If $t \geq m^2$,
\begin{align*} \Pr\left[\text{$\sqrt{m^2 + \chi_{m-1}^2} \geq t$ and $\lvert N\rvert > m$} \right] &\leq \Pr\left[\chi_{m-1}^2 \geq t^2 - m^2 \right]
\\&\leq \exp\left(-\Omega\left(t^2/m - m\right)\right)
\\&\leq \exp\left(-\Omega\left(tm - m\right)\right)
\\&\leq \exp\left(-\Omega\left(t\right) \right).
\end{align*}
Therefore, in all cases, $\Pr[|\typeB| \geq t] \leq O(\exp(-\Omega(t)))$.
 \hfill\qedhere
\end{proof}

\paragraph{Limiting behavior} In the limit $m \to \infty$ under the scaling $U = L/bm$, $dt = 1/m$,~\eqref{eq:drift} is approximated by the stochastic differential equation
\[ dU = \Bigl(-\mu + \frac{1}{2U}\Bigr)dt + \sigma dW, \]
where $\mu = \sqrt{2/\pi}$ and $\sigma = \sqrt{1 - 2/\pi}$ are the statistics of $\lvert N\rvert$, and $W$ is the Wiener process.  The drift pushes $U$ towards the fixed point
\[ U = \frac{1}{2\mu} = \sqrt\frac{\pi}{8} \approx 0.627. \]

\begin{conjecture}[Online threshold conjecture]
\label{online threshold conjecture}
For every $\delta$ and $B$ there exists a sufficiently small $\alpha$ and $\epsilon$ so that every online algorithm fails to find $\vecx \in ([-B, B] \cap \Z)^{n}$ such that $\norm{\matA\vecx}/(\sqrt{m} \norm{\vecx}) < (\sqrt{\pi/8} - \delta) \sqrt{\alpha}/B$ for at least an $\epsilon$ fraction of $\alpha n$ by $n$ matrices $A$ for all sufficiently large $n$.
\end{conjecture}

\subsection{Kernel Rounding}\label{subsection-kernel-rounding}

    \begin{figure}[H]
    \centering
    \fbox{
    \begin{minipage}{0.95\textwidth}
    \begin{center}
    Algorithm $\kernelround$
    \end{center}
\textbf{Input}: $\matA \sim \Norm(0,1)^{m \times n}$ where $m < n$.
\begin{enumerate}
    \item Sample a random $\vecx$ such that $\matA\vecx = \zero$ according to the Haar measure. 
    \item Scale $\vecx$ to have length $\sqrt{\chi^2_n} \cdot \frac{B}{\sqrt{4K \ln^+ B}}$.
    \item Define the rounded vector $\vecz = \round{\vecx}_B$.
\end{enumerate}
\textbf{Output}: The vector $\vecz \in ([-B, B] \cap \Z)^n$.
    \end{minipage}
    }
    \caption{Kernel Rounding Algorithm for $\problem$, as analyzed in \Cref{thm:roundkernel}.}
    \label{fig:kernel-round}
\end{figure}

The rounding function is applied entrywise as
\[ \lceil x \rfloor_B = \begin{cases}
\lceil x \rfloor, &\text{if $\lvert x\rvert \leq B$}, \\
B \sign{x}, &\text{otherwise,}
\end{cases}
\]
where $\round{x}$ denotes rounding $x \in \R$ to the nearest integer (and tie-breaking arbitrarily). We will let
$\ln^+$ denote the function $\max\{\ln, 1\}$. $K \geq 2$ is a (constant) parameter that controls the tradeoff between approximation quality and failure probability in the regime $B = \Omega(n^{1/4})$.

\begin{theorem}
\label{thm:roundkernel}
For $\matA \sim \Norm(0,1)^{m \times n}$, the algorithm $\kernelround$ outputs a vector $\vecz \in ([-B, B] \cap \Z)^n$ such that
\[ \frac{\norm{\matA \vecz}_2}{\norm{\vecz}_2} = O\left( \frac{\sqrt{m K \ln^+ B}}{B} \right),\]
except with probability $2^{-\Omega(m)} + \min\left\{ O\left(2^{-\Omega\left(n(K \ln B)^2/B^4\right)}\right), 2nB^{-K} \right\}$.
\end{theorem}

As the space spanned by the rows of $\matA$ is a random $m$-dimensional subspace of $\R^n$, the marginal distribution of $\vecx$ is identically $\Norm(0,B^2/(4 K \ln^+ B))^n$.  We analyze the conditional distribution of $\langle \veca, \vecx\rangle$ for a single row $\veca$ of $\matA$.

\begin{claim}
\label{claim:kernel}
Suppose $\veca \sim \Norm(0,1)^n$. Conditioned on $\vecx$ and $\langle \veca, \vecx\rangle = 0$, the random variable $\langle \veca, \lceil \vecx \rfloor_B \rangle$ is a centered normal of variance at most $\norm{\{\vecx\}_B}^2$, where $\{x\}_B = x - \lceil x \rfloor_B$ is applied coordinate-wise.
\end{claim}
We emphasize that $\{x\}_B$ can be arbitrarily large when considering $|x| > B$.

\begin{claim}
\label{claim:roundstats}
If $N \sim \Norm(0,1)$, then the random variable
\[\left\{ \frac{B}{\sqrt{4K \ln^+ B}} \cdot  N\right\}_B^2\] 
has mean at most $1/4 + O(B^{-K+2})$ and subexponential norm $O(B^2/(K\ln^+ B))$.
\end{claim}

\begin{proof}[Proof of~\Cref{thm:roundkernel}]
The quantity $\norm{\{\vecx\}_B}^2$ is a sum of $n$ independent
\[\left\{ \frac{B}{\sqrt{4K \ln^+ B}} \cdot  N\right\}_B^2\] 
random variables where $N \sim \Norm(0,1)$. By~\Cref{claim:roundstats} and Bernstein's inequality~\cite[Theorem~2.8.1]{vershynin-book}, $\norm{\{\vecx\}_B}^2$ is $O(n)$ except with probability $O(2^{-\Omega(n(K \ln B)^2/B^4)})$, which is meaningful whenever $B = O(n^{1/4})$. When $B$ is large relative to $n$, by a Gaussian tail bound and a union bound, none of the $n$ entries of $\vecx$ exceed $B$ in absolute value and so $\norm{\{\vecx\}_B}^2 \leq n/4$ except with probability $2n \cdot B^{-K}$.

By~\Cref{claim:kernel}, conditioned on $\vecx$, the entries of $\matA \lceil \vecx \rfloor_B$ are $m$ independent normals of variance $O(n)$ with the same exceptional probability. By Hoeffding's inequality $\norm{\matA \round{\vecx}_B}^2$ itself is bounded by $O(mn)$ except with additional probability $2^{-\Omega(m)}$.

Lastly, we lower bound $\norm{\round{\vecx}_B}$. Each entry of $\lceil \vecx \rfloor_B$ is of magnitude at least $B/\left(4\sqrt{K \ln^+ B}\right)$ with constant probability.  By a large deviation bound,
\[\norm{\lceil \vecx \rfloor_B} \geq \Omega\left( \frac{B \sqrt{n} }{\sqrt{K \ln^+ B} } \right),\]
except with probability $2^{-\Omega(n)}$. By a union bound, the target ratio becomes
\[ \frac{\norm{\matA \round{\vecx}_B}}{\norm{\round{\vecx}_B}} \leq O\left(\frac{\sqrt{mn}}{B \sqrt{n} / \left(\sqrt{K \ln^+ B} \right) }\right) = O\left( \frac{\sqrt{m K \ln^+ B}}{B} \right). \]
\end{proof}

\begin{proof}[Proof of~\Cref{claim:kernel}]
$\langle \veca, \lceil \vecx \rfloor_B\rangle = \langle \veca, \vecx\rangle -\langle \veca, \{\vecx\}_B\rangle$.  Conditioned on $\vecx$ and $\langle \veca, \vecx\rangle = 0$,  $\veca$ is jointly normal and centered (but not independent), so $\langle \veca, \{\vecx\}_B\rangle$ is also centered normal. Conditioning on $\inner{\veca, \vecx} = 0$ implies that $\inner{\veca, \round{\vecx}_B}$ is centered normal too.

To derive its variance, we first calculate the conditional covariances of the entries of $\veca$ given $\vecx$.  Unconditionally, the entries $a_i$ of $\veca$ are independent standard normal.  They decompose as 
\[ a_i = \frac{x_i}{\norm{\vecx}^2} \langle \veca, \vecx\rangle + a_i^\perp, \]
where $a_i^\perp$ is a centered normal independent of $\langle \veca, \vecx\rangle$ given $\vecx$. By independence and expanding out the expression,
\begin{align*}
\Cov(a_i, a_j \mid \vecx, \langle \veca, \vecx \rangle = 0) &= \Cov\left(a_i^\perp, a_j^\perp \mid \vecx, \langle \veca, \vecx \rangle = 0\right) \\
&= \Cov\left(a_i^\perp, a_j^\perp \mid \vecx\right) \\
&= \E\left[ \left(a_i - \frac{x_i}{\norm{\vecx}^2}\langle \veca, \vecx \rangle \right) \left(a_j - \frac{x_j}{\norm{\vecx}^2}\langle \veca, \vecx \rangle \right)\ \Big|\ \vecx \right] \\
&= \one[i = j] -\frac{x_ix_j}{\norm{\vecx}^2}.
\end{align*}
By the variance of sum formula,
\begin{align*}
\Var\left(\langle \veca, \{\vecx\}_B\rangle \mid \vecx, \langle \veca, \vecx \rangle = 0 \right)
&= \sum_i \{x_i\}_B^2 - \sum_{i, j} \{x_i\}_B\{x_j\}_B \cdot \frac{x_ix_j}{\norm{\vecx}^2} \\
&= \norm{\{\vecx\}_B}^2 - \biggl(\frac{\langle \vecx, \{\vecx\}_B\rangle}{\norm{\vecx}}\biggr)^2 \\
&\leq \norm{\{\vecx\}_B}^2.
\end{align*}
Once again, conditioned on $\vecx$ and $\inner{\veca, \vecx} = 0$, this implies that the variance of $\inner{\veca, \round{\vecx}_B}$ is also at most $\norm{\{\vecx\}_B}^2$. 
\hfill\qedhere
\end{proof}

\begin{proof}[Proof of \Cref{claim:roundstats}]
Let $L = B/\sqrt{4K\ln^+ B}$. Then by Cauchy-Schwarz and standard tail bounds,
\begin{align*} 
\E\left[ \{LN\}_B^2 \right]  &\leq \E\left[ \left(LN - \lceil LN \rfloor \right)^2 \cdot \one[LN \leq B] \right] + \E \left[ \left(LN - B \right)^2 \cdot \one[LN > B] \right] \\ 
&\leq \frac14 + \sqrt{\E \left[ (LN - B)^4 \right] } \sqrt{\Pr\left[ LN > B \right]} \\
&\leq \frac14 + O\left(B^2 \exp\left(-B^2/4L^2\right)\right)
\\&\leq \frac{1}{4} + O\left(B^{2-K}\right).
\end{align*}
For every $t > 1/2$,
\[ \Pr\left[ \{LN\}_B^2 \geq t \right] = \Pr\left[ \lvert LN\rvert \geq B + \sqrt{t} \right]
\leq 2 \exp\left( - \frac{(B + \sqrt{t})^2}{2L^2} \right)
= 2 \exp \left( - \left(\frac{B/\sqrt{t} + 1}{L \sqrt{2} }\right)^2 \cdot t \right).
\]
The parenthesized expression squared is at least $1/(2L^2)$, so $\{LN\}_B^2$ is subexponential of norm $O(L^2)$. 
\end{proof}

\section{Hardness}
\subsection{Overlap Gap Property}
\label{sec:hardness}

We study the typical structure of the solution space of~\eqref{eq:sbp} under the parametrization $m = \alpha n$ for fixed $\alpha, \kappa, B$ and large $n$.  Gamarnik et al.~\cite{DBLP:conf/focs/GamarnikK0X22} showed that when the solution $\vecx$ is restricted to the Boolean cube $\{\pm 1\}^n$ the multi-overlap gap property holds in the regime $\alpha \gg \kappa^2 \log 1/\kappa$.  We establish the following extension.

\begin{theorem}
\label{thm:ogp}
For all $\alpha, \kappa, B$ with 
$\kappa \ll 1/B$ and $\alpha \gg (B\kappa)^2 \log 1/B\kappa$ there exists $\beta$ and $r$ such that for all sufficiently large $n$, for all replicas $\vecx_1, \dots, \vecx_r$ such that 
\begin{equation}
\label{eq:ogp}
1 - \beta \leq \frac{\langle\/\vecx_i, \vecx_j\rangle}{\norm{\vecx_i} \norm{\vecx_j}} \leq 1 - \beta + \beta/2r \qquad\text{for all $i \neq j$},    
\end{equation}
at least one $\vecx_i$ fails to satisfy~\eqref{eq:sbp} except with probability $\exp -\Omega(\alpha r n)$.
\end{theorem}

This theorem can be interpreted as evidence for $\problem$ being hard in the regime $\alpha \gg B^2 \kappa^2 \log 1/\kappa$. 

In the algorithmic study of random disordered systems, algorithmic efficiency is predicted by ``replica symmetry'' of the solution space.  Under the replica symmetric model, the uniform distribution over solutions of $\norm{\matA \vecx} \leq \kappa m$ is approximated by a product distribution $P$ over $\{\pm 1\}^n$ with different biases across coordinates.  

Assuming $P$ has sufficient entropy, the overlap $\langle\/\vecx, \vecx'\rangle / \norm{\vecx} \norm{\vecx'}$ between two random solutions ought to be bounded away from 1.  Let $B(t)$ be the ``bouquet'' of vectors $\vecx_1, \dots, \vecx_r$ in which the first $t$ coordinates are sampled identically from $P$ and the rest are sampled independently from $P$.  By the law of large numbers all pairwise overlaps in $B(t)$ will typically concentrate around some value $1 - \beta(t)$.  As $\beta(0)$ is bounded away from zero, $\beta(n)$ is zero, and $\beta$ is $O(1/n)$-Lipschitz, by the intermediate value theorem $\beta(t)$ is bound to hit the interval $(\beta - \beta/2r, \beta)$ for some $t$, contradicting Theorem~\ref{thm:ogp}.

While is difficult to justify the accuracy of the replica symmetric model, and algorithmic easiness may persist even under replica symmetry breaking, the underlying logic can be used to rigorously rule out certain natural classes of algorithms.  In \Cref{sec:onlinehard} we prove that a natural extension of \Cref{thm:ogp}, known as the ensemble (multi) overlap gap property, rules out online algorithms under a certain stability restriction in the claimed regime.  Before stating this extension it is instructive to see the proof of \Cref{thm:ogp}.

\subsection{Proof of the Overlap Gap Property}
\label{sec:ogpproof}

The proof goes by a first moment (annealed) estimate.  Claim~\ref{claim:cover} is used to count the number of forbidden configurations $\vecx_1, \dots, \vecx_r$.  \Cref{claim:determinant}, together with the explicit formula for the multivariate Gaussian PDF, is used to bound the probability of any such configuration simultaneously solving~\eqref{eq:sbp}.  For the parameter regime of interest the expected number of configuration, which is the product of these two numbers, is close to zero.

\begin{claim}
\label{claim:angle}
Assuming $t \leq 1/12B$
 the number of nonzero $B$-bounded points  $\vecx \in \Z^n$ that are within angle $\arctan t$ of some fixed $\vecx_0$ is at most $O((1 + \ln B)/t) \cdot \exp O(n t^2 B^2 \ln 1/t B)$.
\end{claim}

\begin{claim}
\label{claim:ball}
For $\rho \leq 1/2$ an $n$-dimensional ball of radius $\sqrt{\rho n}$ contains at most $\exp O(n \rho \log 1/\rho)$ integer points.
\end{claim}
\begin{proof}
The number of integer points within distance $\sqrt{\rho n}$ of $\vecy$ is at most~\cite[Lemma 1]{mazo-odlyzko}
\[ \exp(s \rho n) \prod_i \sum_k \exp(-s(k - y_i)^2) \]
for every $s > 0$. The summation is minimized at $y_i = 0$ giving an upper bound of
\[ \exp(s \rho n) \biggl( \sum_k \exp(-sk^2)\biggr)^n \leq \exp(s \rho n) \biggl(1 + \frac{2e^{-s}}{1 - e^{-s}}\biggr)^n \]
Setting $s = \ln 1/\rho$ produces the desired bound.  
\end{proof}

\begin{claim}
\label{claim:cover}
Assume $t \leq 1$ and $\norm{\vecx_0} = 1$.  The union of balls with centers $(1 + t)^i \vecx_0$ and radii $\sqrt{5} t(1 + t)^i$ where $i$ ranges over the integers covers all nonzero $\vecx$ within angle  $\arctan t$ of $\vecx_0$.
\end{claim}

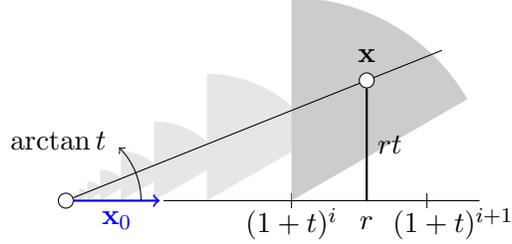
\begin{center}
\begin{tikzpicture}[nstyle/.style={shape=circle, inner sep=2pt,draw, fill=white}]

\node (O) [nstyle] at (0, 0) {};

\foreach \i in {.111, .179, .286, .457, .732, 1.172, 1.875} 
{
  \fill[black!10] (\i, 0) -- (\i, .9 * \i) arc (90:30:.9 * \i) -- cycle;
}

\draw [thin] (3, .1) -- (3, -.1);
\draw [thin] (4.8, .1) -- (4.8, -.1);

\fill[black!20] (3, 0) -- (3, .9 * 3) arc (90:30:.9 * 3) -- cycle;

\draw [thin] (1.3, 0) -- (5.5, 0);
\draw [thick, blue, ->] (O) edge node [below] {$\vecx_0$} (1.25, 0);
\draw [thin] (O) -- (5, 2);

\node (x) [nstyle, label={$\vecx$}] at (4, 1.6) {};
\coordinate (x0) at (4, 0);
\draw [thick] (x) edge node [right] {$rt$} (x0);
  
\node at (3, -.3) {$(1 + t)^i$};
\node at (4, -.3) {$r$};
\node at (5.15, -.3) {$(1 + t)^{i+1}$};

\draw [thin, ->] (1, 0) arc (0:45:1);
\node at (-.1, .8) {$\arctan t$};
\end{tikzpicture}
\end{center}

\begin{proof}
Let $r\vecx_0$, $r > 0$ be the projection of $\vecx$ in the direction of $\vecx_0$.  Take the largest $i$ such that $(1 + t)^i \leq r$.  The distance between $\vecx$ and its projection $r\vecx_0$ is at most $r t$.  By Pythagoras' theorem the distance from $\vecx$ to the center of ball $i$ is at most
\begin{align*} 
\sqrt{r^2 t^2 + \bigl((1 + t)^{i + 1} - (1 + t)^i\bigr)^2} 
&\leq \sqrt{r^2 t^2 + (1 + t)^{2i} t^2} \\
&< (1 + t)^i t \sqrt{(1 + t)^2 + 1}
\end{align*}
because $r < (1 + t)^{i+1}$.  As $t \leq 1$ the square root is at most $5$.
\end{proof}

\begin{proof}[Proof of Claim~\ref{claim:angle}]
Those points are covered by the balls in Claim~\ref{claim:cover}.  We can discard the balls indexed by (negative) $i$ such that $(1 +  \sqrt{5})(1 + t)^i < 1$
as they fit into the unit ball and do not cover any integer points.  We can also discard the balls indexed by $i$ with $(1 + t)^i > (1 + \sqrt{5})B \sqrt{n}$ as they only cover points of magnitude larger than $B \sqrt{n}$.  The largest of the remaining balls has radius at most $(5 + \sqrt{5}) Bt / \sqrt{n} \leq \sqrt{n/2}$ so by Claim~\ref{claim:ball} it contains at most $\exp O(n (B t)^2 \log 1/Bt)$ points.  The number of such balls is $O((1 + \ln B)/t)$.
\end{proof}

\begin{claim}
\label{claim:determinant}
An $r$ by $r$ matrix with diagonal $1$ and off-diagonal entries between $1 - \beta$ and $1 - \beta + \beta/2r$ is positive semidefinite and has determinant at least 
$(\beta/2)^{r-1} (\beta/2 + (1 - \beta)r)$.
\end{claim}
\begin{proof}
This matrix has the form $(1 - \beta)I + \beta J + (\beta / 2r) E$ for some $E$ of infinity-norm at most $1$ so spectral norm at most $r$.  Here $I$ and $J$ are the identity and the all-ones matrices.  The eigenvalues of $(1 - \beta)I + \beta J$ are $\beta$ of multiplicity $r - 1$ and $\beta + (1 - \beta)r$.  
The $(\beta / 2r) E$ term cannot change them by more than $\beta / 2$ each giving the desired bounds.
\end{proof}

\begin{proof}[Proof of Theorem~\ref{thm:ogp}]
Let $C$ be a sufficiently large absolute constant.  For fixed $\vecx_1$, the number of $\vecx_j$ with the required correlation is at most $O_n(1) \cdot \exp O(n \beta B^2 \ln 1/\beta B^2)$ by Claim~\ref{claim:angle}, provided $\beta < 1/C B^2$.   Thus the number of requisite configurations $\vecx_1, \dots, \vecx_r$ is at most $\exp O(n \ln B + n(r - 1) \beta B^2 \ln(1 / \beta B^2))$.  

If~\eqref{eq:sbp} holds for all $\vecx_i$ then all but a $1/4$ fraction of the normalized dot products $\langle \veca, \vecx_i\rangle/\norm{\veca}\norm{\vecx_i}$ are bounded by $4\kappa$ as $a$ ranges over the rows of $\matA$. 
We bound the probability that some fixed subset $S$ of $3rm/4$ correlations are simultaneously at most $4\kappa$.
At least half the $\veca$s then have normalized dot product at most $4\kappa$ with some $r/2$ of the $\vecx_i$s.
The normalized dot products between $\veca$ and $\vecx_1, \dots, \vecx_r$ are 
jointly normal with pairwise correlations between $1 - \beta$ and $1 - \beta + \beta/2r$.  By a union bound the probability that all $r' = r/2$ of them are at most $4\kappa$ is at most
\begin{equation}
\label{eq:det}
\frac{1}{(2\pi)^{r'/2} D^{1/2}} (8\kappa)^{r'} 
\end{equation}
where $D$ is the determinant of the correlation matrix restricted to these $r'$ entries.  By Claim~\ref{claim:determinant} $D$ is at least $(\beta/2)^{r'}$ which gives an overall probability of at most $O(\kappa^2 / \beta)^{r/4}$ for a given $\veca$.  As the rows of $\matA$ are independent the overall probability is at most $O(\kappa^2 / \beta)^{mr/8}$.  By a union bound, the probability that such an $S$ exists is of the same order as there are at most $2^{mr}$ choices for $S$.

The expected number of configurations that solve~\eqref{eq:sbp} is therefore at most
\[ Z = \exp O(n \ln B + n(r - 1) \beta B^2 \ln(1 / \beta B^2)) - (\alpha nr/8) \ln (\beta/\kappa^2) + O(\alpha nr), 
\]
where we replaced $m$ by $\alpha n$.
Setting $\beta = C\kappa^2$, $\log_n Z$ becomes negative in the claimed regime $\alpha \geq C^2 (B\kappa)^2 \log C/B\kappa$ when $r$ is sufficiently large.  By Markov's inequality $\vecx_1, \dots, \vecx_r$ cannot simultaneously solve~\eqref{eq:sbp} with high probability in the limit $n \to \infty$.
\end{proof}

\subsection{Hardness for Online Algorithms}
\label{sec:onlinehard}

The ensemble OGP refers to a correlated ensemble of instances $\matA_1, \dots, \matA_r$.  The condition posits that among all $\vecx_1, \dots, \vecx_r$ satisfying~\eqref{eq:ogp}, at least one pair $\vecx_i$ fails to solve instance $\matA_i$ of~\eqref{eq:sbp}.

The proof of \Cref{thm:ogp} extends readily to establish ensemble OGP for any collection of instances $\mathcal{A}(t) = (\matA_1, \dots, \matA_r)$ in which the first $t$ columns are sampled identically and the other $n - t$ are sampled independently across instances, for any $t$.  The reason is that the value of the determinant $D$ in~\eqref{eq:det} may only increase and all other quantities are preserved.

In contrast, we argue that any ``norm-concentrated'' online algorithm  admits a choice of $t$ for which the resulting solutions do satisfy~\eqref{eq:ogp}.  An algorithm is $(\epsilon, \delta)$-norm concentrated if the norm of the output $\vecx$ is within some $\epsilon \sqrt{n}$-interval of the median except with probability at most $\delta$.  Our algorithm $\cool$ is $(1, 0)$ concentrated.

The analysis relies on another mild assumption that can be enforced.  We say the algorithm is nice if the (unconditional) distribution of the entries of $\vecx$ is symmetric and the entries have independent signs. 

\begin{claim}
\label{claim:nice}
Every online algorithm that on input $\matA$ produces a solution $\vecx$ to~\eqref{eq:sbp} can be made nice without affecting its success probability and the norm of the output.
\end{claim}
\begin{proof}
Choose a random $\vecr \in \{\pm 1\}^m$. 
Multiply the columns of $\matA$ pointwise by $\vecr$, run the algorithm then multiply the solution $\vecx$ pointwise by $\vecr$.  The pointwise multiplications cancel out, preserving the solution.  The solution is symmetrized and its entries become sign-independent.
\end{proof}

\begin{claim}
\label{claim:bouquet}
Assume $\epsilon \leq \min\{ \beta \rho / 16r, \beta \rho^{3/2} / 16B^2r \}$.  For every nice online algorithm that is $(\epsilon, \delta)$-concentrated around norm at least $\sqrt{\rho n}$, there exists a (random) $t$ such that its outputs on inputs $\mathcal{A}(t)$  satisfy~\eqref{eq:ogp} except with probability $rn\delta + r^2n2^{-\Omega(\epsilon^2 n / B^2)}$ for sufficiently large $n$.
\end{claim}
\begin{proof}
$\mathcal{A}(t)$ can be sampled from a common ``stem'' $\matA$ by picking the first $t$ columns in $\matA_i$ as in $\matA$ and the rest independently.  Let $\vecx$ be the output of the algorithm on input $\matA$ and $\vecx^{\leq t}$ (resp., $\matA^{\leq t}$) be its first $t$ entries (resp., columns).  In particular, $\vecx_1(n) = \cdots = \vecx_n(n) = \vecx$.

The sequence of random variables $\langle \vecx_i(t), \vecx_j(t)\rangle - \norm{\vecx^{\leq t}}^2$ is a martingale with respect to the filtration $\matA^{\leq t}$.  The martingale property is a consequence of niceness.  As its increments are $2B^2$-bounded, by Azuma's inequality $\langle \vecx_i(t), \vecx_j(t)\rangle$ is within $\epsilon n$ of $\norm{\vecx^{\leq t}}^2$ given $\matA^{\leq t}$ except with probability $2\exp -\Omega(\epsilon^2 n / B^2)$.  By a union bound, this holds simultaneously for all times and all replica pairs except with $\binom{r}{2} n$ times this probability.

By another union bound, the 2-norms of all replicas $\vecx_i(t)$ at all times $t$ are $\epsilon \sqrt{n}$-concentrated around at least $\sqrt{\rho n}$ except with probability $rn \delta$.  

Assume that the exceptional events do not occur.  Let $f(t) = \norm{\vecx^{\leq t}}^2/\norm{\vecx}^2$.  As $f(0) = 0$, $f(1) = 1$, and $f$ is $B^2/(4\epsilon^2 n)$-Lipschitz, by the intermediate value theorem there must be a (random) time $t$ at which it must be within 
$B^2/(4\epsilon^2 n)$ of $\beta - \beta/4r$.

We argue that $\vecx_1(t), \dots, \vecx_r(t)$ must satisfy~\eqref{eq:ogp}.  For every pair $i \neq j$, 
\begin{align*}
\biggl| f(t) - \frac{\langle\/\vecx_i(t), \vecx_j(t)\rangle}{\norm{\vecx_i(t)} \norm{\vecx_j(t)}} \biggr| 
&\leq 
\norm{\vecx^{\leq t}}^2 \biggl| \frac{1}{\norm{\vecx_i(t)} \norm{\vecx_j(t)}} - \frac{1}{\norm{\vecx}^2}\biggr| + \frac{|\norm{\vecx^{\leq t}}^2 - \langle\/\vecx_i(t), \vecx_j(t)\rangle|}{\norm{\vecx}^2} \\
&\leq B^2n \cdot \frac{2\epsilon}{\rho^{3/2} n} + \frac{\epsilon n}{\rho n} \\
&= \frac{2 \epsilon B^2}{\rho^{3/2}} + \frac{\epsilon}{\rho}.
\end{align*}
By the assumption on $\epsilon$ this is at most $\beta/6r$, so for sufficiently large $n$ all overlaps are within $\beta/4r$ of $\beta - \beta/4r$ as desired.
\end{proof}

\begin{theorem}
\label{thm:onlinehard}
For all $\kappa, \alpha, B$ such that $\kappa \ll 1/B$ and $\alpha \gg (B\kappa)^2 \log 1/B\kappa$ there exists an $\epsilon$ such that
an online $(\epsilon, o(1/n))$-norm concentrated algorithm cannot solve~\eqref{eq:sbp} with probability $1 - o(1/n)$. 
\end{theorem}
\begin{proof}
Assume it does.  By \Cref{claim:nice} the algorithm can be assumed nice.  Let $r$ and $\beta$ be the as promised by \Cref{thm:ogp}.  Let $\rho = (c \alpha \log 1/\kappa) / \log (\alpha \log 1/\kappa)^{-1}$ for a sufficiently small constant $c$.  By \Cref{claim:ball}, \Cref{chi-squared-lower-tail}, and a union bound, no solutions to~\eqref{eq:sbp} of norm at most $\sqrt{\rho n}$ exist except with probability $O(\kappa)^{\alpha n} = o(1/n)$.  By a union bound, all $rn$ instances $\mathcal{A}(1) \cup \cdots \cup \mathcal{A}(n)$ are solved by the algorithm and their solutions have norm at least $\sqrt{\rho n}$ with at least constant probability.  

By \Cref{claim:bouquet}, the outputs of the algorithm on input $\mathcal{A}(t)$ satisfy~\eqref{eq:ogp} and~\eqref{eq:sbp}.
Thus~\eqref{eq:ogp} and~\eqref{eq:sbp} simultaneously hold with constant probability.  This contradicts the ensemble OGP extension of \Cref{thm:ogp}.
\end{proof}

Can the norm concentration assumption be removed from \Cref{thm:onlinehard}?  We believe that some stability restriction on the output is needed for the outputs to satisfy a condition like~\eqref{eq:ogp}.  A general online algorithm may be unstable in the sense that changing even a single input can induce an arbitrarily large change in the norm of its output.

\subsection{Reduction from Continuous Learning with Errors}\label{sec:lattice-lb}

Here, we adapt the proof of \cite{DBLP:vv25} to give a computational lower bound assuming the worst-case hardness of approximate lattice problems. To make the proof simpler (albeit less direct), we reduce from an intermediate problem called Continuous Learning With Errors (CLWE)~\cite{DBLP:conf/stoc/Bruna0ST21}, which is an average-case problem that is known to be as hard as worst-case approximate lattice problems~\cite{DBLP:conf/stoc/Bruna0ST21,DBLP:conf/focs/GupteVV22}.

Let $\sphere^{m-1}$ denote the unit sphere (according to the usual $\ell_2$ metric) in $\R^m$. For this section, we let the $\pmod{1}$ notation denote taking fractional representatives in $[-1/2, 1/2)$ coordinate-wise. We choose these representatives to ensure that $|x \pmod{1}| \leq |x|$ for all $x \in \R$.

\begin{assumption}[Continuous Learning With Errors (CLWE)]\label{clwe}
    For all choices of parameters $n(m) \leq \poly(m)$, $\beta(m) \geq 1/\poly(m)$, and $\gamma(m) = m^{\Omega(1)}$, the following holds. For all p.p.t. adversaries, the following two distributions cannot be distinguished with advantage $\Omega(1)$:
    \[ \left(\matA, \vecs^\top \matA + \vece^\top \pmod{1} \right), \;\;\;  \left(\matA, \vecb^\top \right),\]
    where $\matA \sim \Norm(0,1)^{m \times n}$, $\vecs \sim \gamma \cdot \sphere^{m-1}$, $\vece \sim \Norm(0, \beta^2)^n$, and $\vecb \sim [-1/2,1/2)^n$. 
\end{assumption}

Previous works \cite{DBLP:journals/jacm/Regev09,DBLP:conf/stoc/BrakerskiLPRS13,DBLP:conf/stoc/Bruna0ST21,DBLP:conf/focs/GupteVV22} show the following facts:
\begin{itemize}
    \item \Cref{clwe} holds under the \emph{quantum} worst-case polynomial hardness of the approximate shortest independent vectors problem on lattices (\textsf{SIVP}) and \cite[Corollary 2]{DBLP:conf/focs/GupteVV22}. For $\gamma \geq 2 \sqrt{m}$, this follows from \cite[Corollary 3.2]{DBLP:conf/stoc/Bruna0ST21}, and to get the full range of parameters, from combining \cite[Theorem 1.1]{DBLP:journals/jacm/Regev09} and \cite[Corollary 2]{DBLP:conf/focs/GupteVV22}.
    \item \Cref{clwe} holds under the \emph{classical} worst-case polynomial hardness of the gap shortest vector problem on lattices (\textsf{GapSVP}). This follows from combining \cite[Theorem 1.1]{DBLP:conf/stoc/BrakerskiLPRS13} and \cite[Corollary 2]{DBLP:conf/focs/GupteVV22}.
\end{itemize}
In both cases, \Cref{clwe} holds under the worst-case polynomial hardness of approximately solving lattice problems.

\begin{theorem}\label{lattice-based-lower-bound}
    Under \Cref{clwe}, for all $\varepsilon > 0$ and $B, n \leq \poly(m)$, there does not exist a p.p.t. algorithm for $\problem$ that succeeds with probability at least $2/3$ for
    \[ \kappa = O\left( \frac{1}{B n^{1/2 + \varepsilon}} \right).\]
\end{theorem}
\begin{proof}
    Suppose for contradiction that there exists $\varepsilon > 0$ and $B, n \leq \poly(m)$ with a p.p.t. algorithm $\alg$ succeeding with probability at least $2/3$. We then violate \Cref{clwe} with $\beta = 1/(B \cdot n)$ and $\gamma = n^{\varepsilon/2}$ by the algorithm $\alg'$ described as follows. On input $(\matA, \vecb^\top)$, $\alg'$ runs $\vecx \gets \alg(\matA)$, ensures that $\vecx \in ([-B, B] \cap \Z)^n$ and $\norm{\matA \vecx}_2 < \kappa \norm{\vecx}_2 \sqrt{m}$, and then outputs $1$ if and only if
    \begin{equation}\label{eqn:mod-1-distinguisher} \left|\vecb^\top \vecx \pmod{1} \right| < \frac{1}{4}.
    \end{equation}
    If $\vecx$ is not a solution to $\problem$, then the algorithm outputs $0$.

    We now analyze the performance of $\alg'$. For the ``null'' case of CLWE, where $\vecb \sim [-1/2, 1/2)^n$, since $\vecx \neq 0$ and $\vecx \in \Z^n$, it follows that $\vecb^\top \vecx \pmod{1}$ is distributed uniformly randomly in $[-1/2, 1/2)$. Therefore, \eqref{eqn:mod-1-distinguisher} holds with probability $1/2$ (conditioned on $\vecx$ being a valid solution).

    For the planted case, we know $\vecb^\top = \vecs^\top \matA + \vece^\top \pmod{1}$ for $\vecs \sim \gamma \cdot \sphere^{m-1}$, $\vece \sim \Norm(0, \beta^2)^n$. By spherical symmetry of the Gaussian, we can write $\vecs$ as $\gamma \cdot \vecs'/\norm{\vecs'}_2$ for $\vecs' \sim \Norm(0,1)^n$. Since $\vecx \in \Z^n$, we then have
    \begin{align*} \left| \vecb^\top \vecx \pmod{1} \right| &= \left| \vecs^\top \matA \vecx + \vece^\top \vecx \pmod{1} \right|
    \\&\leq \left| \vecs^\top \matA \vecx \right| + \left| \vece^\top \vecx \right|
    \\&= \frac{\gamma}{\norm{\vecs'}_2} \left| (\vecs')^\top \matA \vecx \right| + \left|\vece^\top \vecx \right|
        \\&= \frac{\gamma}{\norm{\vecs'}_2} \left| v_1 \right| + \left| v_2 \right|,
    \end{align*}
    where $v_1 \sim \Norm(0, \norm{\matA \vecx}_2^2)$ and $v_2 \sim \Norm(0, \beta^2 \norm{\vecx}_2^2)$, by Gaussianity and independence of $\vecs'$ and $\vece$ from $\matA$ and $\vecx$. By standard tail bounds on the (univariate) Gaussian distribution, with probability at least $99/100$, we know that $|v_1| \leq O\left(\norm{\matA \vecx}_2 \right)$ and $|v_2| \leq O(\beta \norm{\vecx}_2)$. Also, by \Cref{chi-squared-lower-tail}, we know that $\norm{\vecs'}_2 \geq \Omega(\sqrt{m})$ with probability $1 - o(1)$. Putting these all together and continuing the inequality, with probability at least $99/100 - o(1)$ (conditioned on $\vecx$ being a valid solution), we have
    \begin{align*} \left| \vecb^\top \vecx \pmod{1} \right| &\leq \frac{\gamma}{\norm{\vecs'}_2} \left| v_1 \right| + \left| v_2 \right|
    \\&\leq O\left( \frac{\gamma \norm{\matA \vecx}_2}{\sqrt{m}}  + \beta \norm{\vecx}_2 \right)
    \\&\leq O\left( \frac{\gamma \kappa \norm{\vecx}_2 \sqrt{m}}{\sqrt{m}}  + \beta \norm{\vecx}_2 \right)
    \\&= O\left(\gamma \kappa + \beta \right) \cdot \norm{\vecx}_2.
    \\&\leq O\left(\gamma \kappa + \beta \right) \cdot B \sqrt{n}.
    \end{align*}
    For $\alg'$ to get advantage $\Omega(1)$, it suffices to show that $(\gamma \kappa + \beta) B \sqrt{n} = o(1)$. We directly have $\beta = o(1/(B \sqrt{n}))$ by the way we have set $\beta$. For the last remaining term, we can plug in our settings of $\gamma$ and $\kappa$ to see that
    \[ \gamma \kappa B \sqrt{n} = n^{\varepsilon/2} \cdot O\left( \frac{1}{B n^{1/2 + \varepsilon}} \right) \cdot B \sqrt{n} = O\left( \frac{1}{n^{\varepsilon/2}} \right) = o(1),\]
    as desired.
\end{proof}

\section{Robust Locality Sensitive Hash Functions}\label{sec:hashing}

\subsection{Definition}
Here, we define robust locality-sensitive hash functions, specialized to the Euclidean norm, following \cite{DBLP:conf/innovations/BoyleLV19}.

\begin{definition}[Robust Locality Sensitive Hash Functions]\label{def-of-our-hash-primitive} For natural numbers $n$ and $m = m(n) < n$, let $\cX_n \subseteq \R^n$ and $\cY_n \subseteq \R^m$ be finite sets. A \emph{robust locality sensitive hash function} with approximation factors $\alpha = \alpha(n), \beta = \beta(n)$ consists of p.p.t. algorithms $(\keygen, \hash)$ with the following syntax:
\begin{itemize}
    \item $\keygen(1^n) \to \key$. This algorithm is randomized and outputs some key $\key \in \{0,1\}^{\poly(n)}$.
    \item $\hash : \{0,1\}^{\poly(n)} \times \cX_n \to \cY_n $. This algorithm is deterministic. As shorthand, we will write $\hash_{\key} : \cX_n \to \cY_n$ to denote the hash function $\hash(\key, -)$ for fixed key $\key \in \{0,1\}^{\poly(n)}$.
\end{itemize}
Moreover, we require the following three properties:
\begin{enumerate}
    \item \textbf{Compression}\label{item:compression}: We have $|\cY_n| \leq \frac{1}{2} |\cX_n|$. That is, the function $\hash_{\key} : \cX_n \to \cY_n$ is compressing by at least a factor of $2$ (typically, significantly more).
    \item \textbf{Statistical Non-Expansion}\label{item:statistically-not-expanding}:
    \[ \Pr_{\key \sim \keygen(1^n)}\left[\exists \vecx_1, \vecx_2 \in \cX_n : \norm{\hash_{\key}(\vecx_1) - \hash_{\key}(\vecx_2)}_2 > \alpha \cdot \norm{\vecx_1 - \vecx_2}_2 \right] = \negl(n). \]
    \item \textbf{Computational Non-Contraction}\label{item:computationally-not-contracting}: For all p.p.t. adversaries $\cA$, 
    \[ \Pr_{\key \sim \keygen(1^n)}\left[(\vecx_1, \vecx_2) \gets \cA\left(1^n, \key \right) : \begin{array}{c}\vecx_1, \vecx_2 \in \cX_n\; \land
    \\\norm{\hash_{\key}(\vecx_1) - \hash_{\key}(\vecx_2)}_2 < \beta \norm{\vecx_1 - \vecx_2}_2
    \end{array}\right] = \negl(n), \]
    where the probability is also taken over the internal randomness of $\cA$.
\end{enumerate}
We refer to the quantity $\xi = \alpha/\beta$ as the \emph{distortion} of the hash function.
\end{definition}

We note that this definition is, in particular, stronger than a collision-resistant hash function.\footnote{The only syntactic difference is that the codomain here is not expressed as $\{0,1\}^\ell$ for some $\ell$, but rather some efficiently recognizable finite set $\cY_n$. By considering a direct binary encoding of $([-r, r] \cap \gamma \Z)^m \supseteq \cY_n$, one can make the codomain $\{0,1\}^\ell$ with a slight loss in parameters.} To see this, note that if $\vecx_1, \vecx_2 \in \cX_n$ form a collision, then $\hash_{\key}(\vecx_1) = \hash_{\key}(\vecx_2)$ while $\norm{\vecx_1 - \vecx_2}_2 > 0$, which violates computational non-contraction. Therefore, computational assumptions are necessary to achieve the definition.

\subsection{Preliminaries}

Let $\ball_m(r, \vecy) \subseteq \R^m$ denote the ball of radius $r$ (according to the usual $\ell_2$ norm) in dimension $m$ centered at $\vecy \in \R^m$. If $\vecy$ is ommitted, it is taken to be the all $0$s vector.
\begin{lemma}[E.g., Lemma 16 in \cite{DBLP:conf/crypto/KirchnerF15}]\label{integer-point-ell-2-bound}
For all $m \in \N$ and $r \in \R_{>0}$, we have the bound
\[ \left| \ball_m(r) \cap \Z^m \right| \leq  \vol\left(\ball_m\left(r + \frac{\sqrt{m}}{2} \right) \right)\]
\end{lemma}
\begin{proof}
    For all $\vecy \in  \ball_m(r) \cap \Z^m$, consider the (open) cube $\vecy + (-1/2, 1/2)^m$. Note that since $\vecy \in \Z^m$, all such cubes are disjoint. Since all cubes have volume $1$, we have
    \begin{align*} \left|\ball_m(r) \cap \Z^m \right| &= \vol\left( \bigsqcup_{\vecy \in \ball_m(r) \cap \Z^m} 
    \left( \vecy + (-1/2, 1/2) \right)^m \right)
    \\&\leq \vol\left( \bigcup_{\vecy \in \ball_m(r) \cap \Z^m} \ball_m\left( \frac{\sqrt{m}}{2}, \vecy \right) \right)
    \\&\leq \vol\left( \ball_m\left(r + \frac{\sqrt{m}}{2} \right) \right),
    \end{align*}
as desired.
\end{proof}

\begin{corollary}\label{exp-num-of-integer-points-in-ball}
    For all $m \in \N$ and all $\gamma, r \in \R_{>0}$, we have the bound
\[ \left| \ball_m(r) \cap \gamma \Z^m \right| \leq  \left( \frac{r \sqrt{2\pi e}}{\gamma \sqrt{m}}  + \frac{\sqrt{2 \pi e}}{2} \right)^{m}. \]
\end{corollary}
\begin{proof}
    By a scaling argument (by $1/\gamma$), we know
    \[  \left| \ball_m(r) \cap \gamma \Z^m \right| =  \left| \ball_m(r/\gamma) \cap \Z^m \right|.\]
    Plugging in \Cref{integer-point-ell-2-bound},
     \[  \left| \ball_m(r) \cap \gamma \Z^m \right| \leq \vol\left(\ball_m\left(\frac{r}{\gamma} + \frac{\sqrt{m}}{2} \right) \right).\]
    By using a standard bound that
    \[ \vol(\ball_m(R)) \leq \left( \sqrt{\frac{2 \pi e}{ m}} \right)^m R^m,\]
    we get
    \[ \left| \ball_m(r) \cap \gamma \Z^m \right| \leq \left( \sqrt{\frac{2 \pi e}{ m}} \right)^m \cdot \left( \frac{r}{\gamma} + \frac{\sqrt{m}}{2} \right)^m =  \left( \frac{r \sqrt{2 \pi e}}{\gamma \sqrt{m}} + \frac{\sqrt{2 \pi e}}{2} \right)^m,\]
    as desired.
\end{proof}

For a matrix $\matA$, we let $\norm{\matA}_2$ denote the standard spectral norm of $\matA$, i.e., the largest singular value of $\matA$.

\begin{lemma}[As in \cite{rudelson2010non}]\label{lemma:gaussian-singular-value}
For all $t > 0$, we have
\[ \Pr_{\matA \sim \Norm(0,1)^{m \times n}}\left[\norm{\matA}_2 \leq \sqrt{m} + \sqrt{n} + t \right] \geq 1 - 2 e^{-t^2/2}.\]
In particular, for $n > m$ and setting $t = \sqrt{n}$, we have
\[ \Pr_{\matA \sim \Norm(0,1)^{m \times n}}\left[\norm{\matA}_2 \leq 3\sqrt{n} \right] \geq 1 - 2 e^{-n/2}.\]
\end{lemma}

\subsection{Construction}
\begin{theorem}\label{main-construction}
    Suppose that $\problem$ is hard for parameters $n,m,B, \kappa$. Then, for $r = 4Bn/\sqrt{m}$ and $\gamma = \kappa/(2\sqrt{m})$, there exists a universal constant $C$ and a robust locality sensitive hash function for $\cX_n = ([0, B] \cap \Z)^n$ and $\cY_n = \ball_m(r) \cap \gamma \Z^m$ with parameters
    \begin{align*}
        \alpha = 4 \sqrt{\frac{n}{m}}\;\;\;\text{and}\;\;\; \beta = \frac{\kappa}{2},
    \end{align*}
    as long as
       \[ (B+1)^n > \left( \frac{C Bn}{\kappa \sqrt{m}} \right)^m. \]
    In particular, the distortion is
    \[ \xi = \frac{\alpha}{\beta} = \frac{8}{\kappa} \sqrt{\frac{n}{m}}.\]
\end{theorem}

We will use the function $\floor{\cdot}_{\gamma} : \R^m \to \gamma \Z^m$ to denote coordinate-wise rounding down to the nearest multiple of $\gamma$. 

\begin{proof}[Proof of \Cref{main-construction}]
    \sloppy{We construct a robust locality sensitive hash function as described~in~\Cref{fig:construction}.}

    \begin{figure}[H]
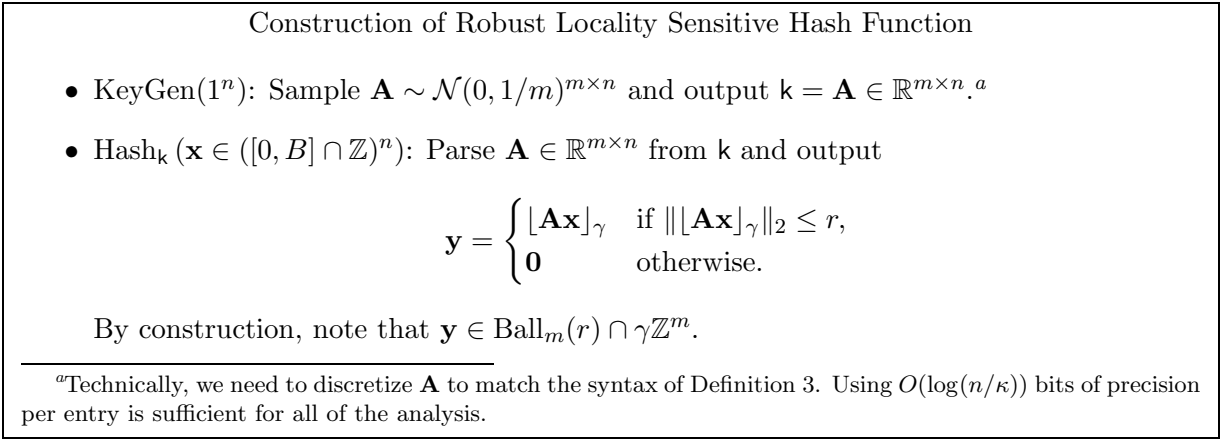

    \centering
    \fbox{
    \begin{minipage}{0.95\textwidth}
    \begin{center}
    Construction of Robust Locality Sensitive Hash Function
    \end{center}
    \begin{itemize}
        \item $\keygen(1^n)$: Sample $\matA \sim \Norm(0, 1/m)^{m \times n}$ and output $\key = \matA \in \R^{m \times n}$.\footnote{Technically, we need to discretize $\matA$ to match the syntax of \Cref{def-of-our-hash-primitive}. Using $O(\log(n/\kappa))$ bits of precision per entry is sufficient for all of the analysis.}
        \item $\hash_{\key}\left(\vecx \in ([0, B] \cap \Z)^n\right)$: Parse $\matA \in \R^{m \times n}$ from $\key$ and output
        \[ \vecy = \begin{cases} \floor{\matA \vecx}_{\gamma} & \text{if } \norm{\floor{\matA \vecx}_{\gamma}}_2 \leq r,\\
        \zero & \text{otherwise.}\end{cases}\]
        By construction, note that $\vecy \in \ball_m(r) \cap \gamma \Z^m$.
    \end{itemize}
    
    \end{minipage}
    }
    \caption{The construction of the robust locality sensitive hash function, as used in \Cref{main-construction}. See \Cref{def-of-our-hash-primitive} for the syntax of robust locality sensitive hash functions.}
    \label{fig:construction}
\end{figure}

For simplicity, we assume for now that it always holds that $\norm{\floor{\matA \vecx}_{\gamma}}_2 \leq r$, i.e., that $\hash_{\key}(\vecx) = \floor{\matA \vecx}_{\gamma}$ for all $\vecx \in \cX_n$. Later in \Cref{output-norm-bound}, we show that with $1 - \negl(n)$ probability over $\matA$, this indeed holds, allowing us to add a $\negl(n)$ term back into the proofs of \Cref{item:computationally-not-contracting,item:statistically-not-expanding,item:compression} in \Cref{def-of-our-hash-primitive}. We prove that \Cref{item:computationally-not-contracting}, \Cref{item:statistically-not-expanding}, and \Cref{item:compression} hold, in that order.

We begin by showing \Cref{item:computationally-not-contracting}.
\begin{claim}[Computational Non-Contraction]
    \Cref{item:computationally-not-contracting} in \Cref{def-of-our-hash-primitive} holds with $\beta = \kappa/2$. That is, assuming the hardness of $\problem$, no p.p.t. algorithm can output $\vecy, \vecz \in \cX_n$ such that $ \norm{\hash_{\key}(\vecy) - \hash_{\key}(\vecz)}_{2} < \kappa/2 \cdot \norm{\vecy - \vecz}_2$ with non-negligible probability.
\end{claim}
\begin{proof}
    The reduction from $\problem$ is direct. For an instance of $\problem$ with matrix $\matA' \sim \Norm(0,1)^{m \times n}$, let $\key = \matA = \frac{1}{\sqrt{m}} \matA' \sim \Norm(0, 1/m)^{m \times n}$ be the key for the robust locality sensitive hash function.
    
    For any such violating pair $\vecy, \vecz$ (where it must be the case that $\vecy \neq \vecz)$, the reduction outputs $\vecx = \vecy - \vecz \in ([-B, B] \cap \Z)^n$. Let $\vece_1 = \matA \vecy - \floor{\matA \vecy}_{\gamma} \in [0, \gamma)^m$, $\vece_2 = \matA \vecz - \floor{\matA \vecz}_{\gamma} \in [0, \gamma)^m$.
    We have
        \begin{align*} \frac{1}{\sqrt{m}} \norm{\matA' \vecx}_2 = \norm{\matA \vecx}_{2} &= \norm{\matA \vecy - \matA \vecz}_{2}
        \\&= \norm{\floor{\matA \vecy}_{\gamma} + \vece_1 -(\floor{\matA \vecz}_{\gamma} + \vece_2)}_{2}
    \\&\leq \norm{\floor{\matA \vecy}_{\gamma} -\floor{\matA \vecz}_{\gamma}}_{2} +  \norm{\vece_1 - \vece_2}_{2}
    \\&= \norm{\hash_{\key}(\vecy) - \hash_{\key}(\vecz)}_{2} + \norm{\vece_1 - \vece_2}_{2}
    \\&\leq \norm{\hash_{\key}(\vecy) - \hash_{\key}(\vecz)}_{2} + \gamma \sqrt{m}
    \\&<  \frac{\kappa}{2} \cdot \norm{\vecy - \vecz}_{2} + \frac{\kappa}{2}
    \\&\leq \frac{\kappa}{2} \cdot \norm{\vecy - \vecz}_{2} + \frac{\kappa}{2} \cdot \norm{\vecy - \vecz}_{2}
    \\&= \kappa \cdot \norm{\vecy - \vecz}_{2},
    \end{align*}
    where the last inequality comes from the fact that $\vecy \neq \vecz$ and $\vecy, \vecz \in \Z^n$. Multiplying both sides by $\sqrt{m}$ gives
    \[ \norm{\matA' \vecx}_2 < \kappa \sqrt{m} \cdot \norm{\vecy - \vecz}_2 = \kappa \sqrt{m} \cdot \norm{\vecx}_2,\]
solving $\problem$ with instance $\matA'$.
\end{proof}

Next, we show \Cref{item:statistically-not-expanding}.
\begin{claim}[Statistical Non-Expansion]\label{stat-non-expansion-proof}
\Cref{item:statistically-not-expanding} in \Cref{def-of-our-hash-primitive} holds with
\[ \alpha = 4 \sqrt{\frac{n}{m}}.\]
More explicitly, with probability at least $1 - 2e^{-n/2}$, for all $\vecy, \vecz \in \cX_n$, 
\[ \norm{\hash_{\key}(\vecy) - \hash_{\key}(\vecz)}_{2} < 4 \sqrt{\frac{n}{m}} \cdot \norm{\vecy - \vecz}_2.\]
\end{claim}
\begin{proof}
    Up to rounding concerns, this is equivalent to upper bounding $\norm{\matA \vecx}_2/\norm{\vecx}_2$ over all $\vecx \in [-B, B]^n \cap \Z^n \setminus \{\zero\}$. We can bound this directly by the spectral norm of $\matA$:
    \[ \max_{\vecx \in ([-B, B] \cap \Z)^n \setminus \{\zero\}} \frac{\norm{\matA \vecx}_2}{\norm{\vecx}_2} \leq \sup_{\vecx \in \R^n \setminus \{\zero\}} \frac{\norm{\matA \vecx}_2}{\norm{\vecx}_2} = \norm{\matA}_2. \]
    Since $\sqrt{m} \cdot \matA \sim \Norm(0,1)^{m \times n}$, by \Cref{lemma:gaussian-singular-value}, it follows that with probability at least $1 - 2 e^{-n/2}$, 
\[ \max_{\vecx \in ([-B, B] \cap \Z)^n \setminus \{\zero\}} \frac{\norm{\matA \vecx}_2}{\norm{\vecx}_2} \leq \norm{\matA}_2 \leq 3 \sqrt{\frac{n}{m}}.\]

Finally, we incorporate the rounding. For $\vecy = \vecz$, the desired inequality is trivially true, so we assume $\vecy \neq \vecz$. Let $\vece_1 = \matA \vecy - \floor{\matA \vecy}_{\gamma} \in [0, \gamma)^m$ and $\vece_2 = \matA \vecz - \floor{\matA \vecz}_{\gamma} \in [0, \gamma)^m$. For $\vecx = \vecy - \vecz \neq \zero$, we have
\begin{align*}
    \norm{\hash_{\key}(\vecy) - \hash_{\key}(\vecz)}_{2} &= \norm{\floor{\matA \vecy}_{\gamma} - \floor{\matA \vecz}_{\gamma}}_2
    \\&= \norm{\matA \vecy - \vece_1 - (\matA \vecz - \vece_2)}_2
    \\&\leq \norm{\matA \vecy - \matA \vecz}_2 + \norm{\vece_1 - \vece_2}_2 
    \\&= \norm{\matA \vecx}_2 + \norm{\vece_1 - \vece_2}_2 
    \\&\leq 3 \sqrt{\frac{n}{m}} \cdot \norm{\vecx}_2  + \gamma \sqrt{m}
    \\&\leq 3 \sqrt{\frac{n}{m}} \cdot \norm{\vecx}_2  + \gamma \sqrt{m} \cdot \norm{\vecx}_2
    \\&= \left(3 \sqrt{\frac{n}{m}}  + \frac{\kappa}{2} \right) \cdot \norm{\vecx}_2,
\end{align*}
where the last inequality comes from the fact that $\vecx \neq 0$. Since $m < n$ and $\kappa < 1$, we can continue the above inequality to see that
\begin{align*}
\norm{\hash_{\key}(\vecy) - \hash_{\key}(\vecz)}_{2} &\leq \left( 3 \sqrt{\frac{n}{m}} + \frac{\kappa}{2} \right) \cdot \norm{\vecx}_2
\\&\leq \left( 3 \sqrt{\frac{n}{m}} + \frac{1}{2} \right) \cdot \norm{\vecx}_2
\\&< 4 \sqrt{\frac{n}{m}} \cdot \norm{\vecx}_2,
\end{align*}
as desired.
\end{proof}
As a consequence of \Cref{stat-non-expansion-proof}, we have the following:
\begin{claim}[Output Norm Bound]\label{output-norm-bound}
    With probability at least $1 - (2B+1)^{-n}$, for all $\vecy \in \cX_n$, 
\[ \norm{\floor{\matA \vecy}_{\gamma}}_{2} < \frac{4Bn}{\sqrt{m}}.\]
\end{claim}
\begin{proof}
    This follows from \Cref{stat-non-expansion-proof} by setting $\vecz = \zero$ and using the bound $\norm{\vecy}_2 \leq B \sqrt{n}$.
\end{proof}

Lastly, we show \Cref{item:compression}.

\begin{claim}[Compression]
    \Cref{item:compression} holds in \Cref{def-of-our-hash-primitive}. More specifically, the function $\hash_{\key}$ is compressing (by a factor of at least $2$) if
    \[(B+1)^n > 2 \left( \frac{Bn \sqrt{128 \pi e}}{\kappa \sqrt{m}}  + \frac{\sqrt{2 \pi e}}{2} \right)^{m}.\]
\end{claim}
\begin{proof}
    Recall that $\cX_n = ([0, B] \cap \Z)^n$ and $\cY_n = \ball_m(r) \cap \gamma \Z^m$. The domain has cardinality
    \[ \left| \cX_n \right| =  \left|([0, B] \cap \Z)^n \right| = (B+1)^n.\]
    By \Cref{exp-num-of-integer-points-in-ball}, we know the codomain $\cY_n$ has cardinality
    \[ \left| \cY_n \right| = \left|\ball_m(r) \cap \gamma \Z^m \right| \leq \left( \frac{r \sqrt{2\pi e}}{\gamma \sqrt{m}}  + \frac{\sqrt{2 \pi e}}{2} \right)^{m}.\]
    Plugging in our value of $r$ from \Cref{output-norm-bound}, this becomes
    \[ \left| \cY_n \right| \leq \left( \frac{Bn \sqrt{32 \pi e}}{\gamma m}  + \frac{\sqrt{2 \pi e}}{2} \right)^{m}.\]
    Therefore, since $\gamma = \kappa/(2\sqrt{m})$, for the function to be compressing by a factor of at least $2$, it suffices that
    \[ (B+1)^n > 2 \left( \frac{Bn \sqrt{128 \pi e}}{\kappa \sqrt{m}}  + \frac{\sqrt{2 \pi e}}{2} \right)^{m}.\]
\end{proof}

\end{proof}

\section{Statistical Threshold}\label{sec:statistical-threshold}

In this section, we describe the statistical threshold for $\problem$ by classifying when the expected number of solutions $\vecx$ is at least $1$. For more precise bounds that take into account more than the first moment (for related variants of this problem), we direct the reader to the works by Aubin, Perkins and Zdeborov\'{a}~\cite{Aubin_2019}, Perkins and Xu~\cite{DBLP:conf/stoc/PerkinsX21}, and Abbe, Li and Sly~\cite{abbe2021proofcontiguityconjecturelognormal}, which confirm that the first moment bound accurately gives a sharp statistical threshold for binary symmetric perceptron models.

We begin by stating a lower tail bound on the $\chi^2_{m}$ distribution.
\begin{lemma}\label{chi-squared-lower-tail}
    There exist universal constants $C_1, C_2 > 0$ such that for all $\kappa \leq 1/2$ and $m \in \N$,
    \[ (C_1 \cdot \kappa)^m \leq \Pr_{Z \sim \chi^2_m}\left[Z \leq \kappa^2 m \right] \leq (C_2 \cdot \kappa)^m. \]
\end{lemma}
\begin{proof}
    We first show the lower bound. Recall that we can characterize $Z \sim \chi^2_m$ by 
    \[ Z = \sum_{i \in [m]} X_i^2 \]
    for i.i.d. $X_i \sim \Norm(0, 1)$. If for all $i \in [m]$ it holds that $|X_i| \leq \kappa$, then
    \[ Z = \sum_{i \in [m]} X_i^2 \leq \kappa^2 m. \]
    Therefore,
    \[\Pr_{Z \sim \chi^2_m}\left[Z \leq \kappa^2 m \right] \geq  \Pr\left[ \forall i \in [m],\; |X_i| \leq \kappa \right] = \Pr_{X \sim \Norm(0,1)}\left[|X| \leq \kappa \right]^m \geq (C_1 \cdot \kappa)^m,\]
    for some universal constant $C_1$, where the right-hand-most inequality holds since the measure of $\Norm(0,1)$ is at least $e^{-1/8}/\sqrt{2\pi}$ on $[-1/2, 1/2]$ and $\kappa \leq 1/2$.

    We now show the upper bound. Using the moment-generating function of the $\chi^2_m$ distribution and the proof of Cram\'er's theorem, for all $\kappa < 1$, we have
    \begin{align*} \Pr_{Z \sim \chi^2_m}\left[Z \leq \kappa^2 m \right] &\leq \exp \left(m \cdot \frac{\ln(\kappa^2) - \kappa^2 + 1}{2} \right)
    \\&= \exp \left( m \cdot \left( \ln(\kappa) - \frac{\kappa^2}{2} + \frac{1}{2} \right) \right)
    \\&= \left( \kappa \cdot \exp \left(- \frac{\kappa^2}{2} + \frac{1}{2} \right) \right)^m
    \\&\leq (\kappa \cdot C_2)^m
    \end{align*}
    for some universal constant $C_2$, since the function $\exp(-\kappa^2/2 + 1/2)$ lies in the interval $[e^{3/8}, e^{1/2}]$ for all $\kappa \in [0, 1/2]$.
\end{proof}

We proceed to compute the statistical bound. For $\matA \sim \Norm(0,1)^{m \times n}$ and fixed $\vecx \in [-B, B]^n \cap \Z^n$, let $I_{\vecx}$ be the indicator random variable given by
\[ I_{\vecx} := \one\left[ \norm{\matA \vecx}_2 < \kappa \sqrt{m} \cdot \norm{\vecx}_2 \right] \in \{0, 1\}. \]
By linearity of expectation, we know that the expected number of solutions is given by
\begin{align*} \E \left[ \left| \left\{ \vecx \in [-B, B]^n \cap \Z^n : \norm{\matA \vecx}_2  < \kappa \sqrt{m} \cdot \norm{\vecx}_2\right\} \right| \right] &= \sum_{\vecx \in [-B, B]^n \cap \Z^n} \E \left[I_{\vecx} \right]
\\&= \sum_{\vecx \in [-B, B]^n \cap \Z^n} \Pr\left[ \norm{\matA \vecx}_2 < \kappa \sqrt{m} \cdot \norm{\vecx}_2 \right].
\end{align*}
For $\matA \sim \Norm(0,1)^{m \times n}$ and fixed $\vecx \in [-B, B]^n \cap \Z^n \setminus \{0^n\}$, by standard properties of the Gaussian distribution, the distribution of $\matA \vecx$ is $\Norm(0, \norm{\vecx}_2^2)^m$. Therefore, $\matA \vecx / \norm{\vecx}_2 \sim \Norm(0,1)^m$, so
\[ \frac{\norm{\matA \vecx}^2_2}{\norm{\vecx}_2^2} \sim \chi^2_m. \]
It follows that
\[ \Pr\left[ \norm{\matA \vecx}_2 < \kappa \sqrt{m} \cdot \norm{\vecx}_2 \right] = \Pr\left[ \frac{\norm{\matA \vecx}_2^2}{\norm{\vecx}_2^2} < \kappa^2 m \right] = \Pr_{Z \sim \chi^2_m}\left[Z < \kappa^2 m \right]. \]
For $\kappa \leq 1/2$, by \Cref{chi-squared-lower-tail}, we have
\[\Pr_{Z \sim \chi^2_m}\left[Z < \kappa^2 m \right] = \Theta \left( \kappa \right)^m. \]
Since $\vecx = 0^n$ is never a solution, it follows that
\begin{align*} \E \left[ \left| \left\{ \vecx \in [-B, B]^n \cap \Z^n : \norm{\matA \vecx}_2  < \kappa \sqrt{m} \cdot \norm{\vecx}_2\right\} \right| \right] &= \left( (2B+1)^n - 1\right) \cdot \Pr_{Z \sim \chi^2_m}\left[Z < \kappa^2 m \right]
\\&= \left( (2B+1)^n - 1\right) \cdot \Theta(\kappa)^m.
\end{align*}
Setting this to $1$ gives the statistical threshold
\[ \kappa = \Theta\left( (2B+1)^{-n/m} \right).\]

\ifnum\anon=0
\paragraph{Acknowledgements.} We thank Avi Wigderson for advice, and we are grateful to Aayush Jain, Thuy-Duong Vuong, and Justin Y. Chen for useful discussions.
\fi

\ifnum\lncs=1
\bibliographystyle{splncs04}
\else
    \ifnum\iacrcc=1
    \else
        \bibliographystyle{alpha}
    \fi
\fi
\bibliography{refs}

\end{document}